\documentclass[sigconf]{acmart}

\usepackage{balance}
\usepackage{enumitem,amsmath,amssymb,graphicx}
\usepackage{amsthm}
\usepackage{xcolor}
\usepackage{bbm}
\usepackage{hyperref}
\usepackage[ruled,vlined]{algorithm2e}
\newcommand\bbR{\ensuremath{\mathbb{R}}} 
\DeclareMathOperator*{\tr}{tr} 

\newcommand{\1}{\mathbb{I}} 

\newtheorem{theorem}{Theorem}
\newtheorem{lemma}[theorem]{Lemma}

\def\BibTeX{{\rm B\kern-.05em{\sc i\kern-.025em b}\kern-.08emT\kern-.1667em\lower.7ex\hbox{E}\kern-.125emX}}
    
\copyrightyear{2020}
\acmYear{2020}
\setcopyright{acmlicensed}\acmConference[SPAA '20]{Proceedings of the 32nd ACM Symposium on Parallelism in Algorithms and Architectures}{July 15--17, 2020}{Virtual Event, USA}
\acmBooktitle{Proceedings of the 32nd ACM Symposium on Parallelism in Algorithms and Architectures (SPAA '20), July 15--17, 2020, Virtual Event, USA}
\acmPrice{15.00}
\acmDOI{10.1145/3350755.3400210}
\acmISBN{978-1-4503-6935-0/20/07}

\renewcommand\footnotetextcopyrightpermission[1]{}

\begin{document}

\fancyhead{}

\title{Spectral Lower Bounds on the I/O Complexity of\\Computation Graphs}

\author{Saachi Jain}
\email{saachi@cs.stanford.edu}
\affiliation{%
  \institution{Stanford University}
  \city{Stanford}
  \state{CA}
  \country{USA}
}
\author{Matei Zaharia}
\email{matei@cs.stanford.edu}
\affiliation{%
  \institution{Stanford University}
  \city{Stanford}
  \state{CA}
  \country{USA}
}

%
\renewcommand{\shortauthors}{Jain and Zaharia, et al.}

%
\begin{abstract}
  We consider the problem of finding lower bounds on the I/O complexity of arbitrary computations in a two level memory hierarchy. Executions of complex computations can be formalized as an evaluation order over the underlying computation graph. However, prior methods for finding I/O lower bounds leverage the graph structures for specific problems (e.g matrix multiplication) which cannot be applied to arbitrary graphs. In this paper, we first present a novel method to bound the I/O of any computation graph using the first few eigenvalues of the graph's Laplacian. We further extend this bound to the parallel setting. This spectral bound is not only efficiently computable by power iteration, but can also be computed in closed form for graphs with known spectra. We apply our spectral method to compute closed-form analytical bounds on two computation graphs (the Bellman-Held-Karp algorithm for the traveling salesman problem and the Fast Fourier Transform), as well as provide a probabilistic bound for random Erd\H{o}s R\'enyi graphs. We empirically validate our bound on four computation graphs, and find that our method provides tighter bounds than current empirical methods and behaves similarly to previously published I/O bounds.

\end{abstract}

%
%


%

%
\maketitle

\section{Introduction}
\let\thefootnote\relax\footnotetext{This is the full version of the paper appearing in the ACM Symposium on Parallelism in Algorithms and Architectures (SPAA), 2020}
Many important applications are bottlenecked not by processing speeds, but by I/O cost: the speed to transfer data items between fast memory (e.g., registers or the CPU cache) and slow memory (e.g., RAM or disk).
There has thus been considerable interest in designing I/O efficient algorithms and in understanding I/O lower bounds~\cite{elangodissertation, ballard2012graph, irony2004communication, redbluepebble}.

Past work on I/O lower bounds has largely focused on finding bounds for specific algorithms, such as matrix multiplication or the Fast Fourier Transform~\cite{elangodissertation, ballard2012graph, irony2004communication, redbluepebble}. However, these approaches leverage properties specific to the tasks at hand, and do not translate across tasks. In this paper, we explore methods that can be applied to \textit{arbitrary} computations and can be computed efficiently in an automatic fashion. Such generic bounds can be used to characterize the I/O cost of computations that are too complex to analyze by hand. Our method also provides a new approach for finding closed form theoretical bounds on computation graphs as long as the Laplacian eigenvalues (or bounds on these values) are known.

We approach the problem of minimizing I/O for an arbitrary computation as finding an optimal evaluation order on the underlying directed computation graph. In a computation graph, each vertex represents a single operation: the parents of the vertex indicate the operands of the operation. We assume a two-level memory architecture with a fixed amount of fast memory and infinite slow memory: I/O is incurred when transferring data between fast and slow memory (Section \ref{memorymodel}).

We present a novel method to provide lower bounds on the I/O for any computation graph using the eigenvalues of the graph Laplacian (Section \ref{spectral}). We further extend this bound to the parallel setting. This spectral bound is efficiently computable and can be applied to arbitrarily large and complex graphs. For graphs with known spectra, this bound can also be computed in closed form. We compute closed form bounds for two computation graphs: the Bellman-Held-Karp algorithm for the traveling salesman problem (TSP) as well as the Fast Fourier Transform (FFT). In the process, we also present a novel result on the multiplicity of the eigenvalues of the butterfly graph, which we use to complete the bound for the FFT. We find that spectral bound for the FFT graph is at most a factor of $(1/\log M)$ weaker than the previously published asymptotically tight bound (where $M$ is the size of fast memory), which was computed via direct inspection of the butterfly graph using $S$ partitions \cite{redbluepebble}. We additionally present a probabilistic bound for random Erd\H{o}s R\'enyi graphs.

We evaluate our method empirically by computing lower bounds for four types of computation graphs: the Fast Fourier Transform, matrix multiplication (naive and Strassen), and the Bellman-Held-Karp algorithm (Section \ref{evaluation}). We find that our bounds are tighter than current automatic methods ~\cite{elangostcut} and behave similarly to published analytical bounds.
\section{Related Work}
\label{relatedwork}
Hong and Kung first framed the problem of I/O complexity as the ``red-blue pebble game" and used it to prove several bounds~\cite{redbluepebble}. The game represents slow memory as an infinite pool of blue pebbles and fast memory as a finite set of red pebbles. An evaluation then corresponds to pebbling each vertex of the graph according to the game; I/O is incurred when placing a red pebble on top of a blue pebble (reading from slow memory) or vice-versa (writing to slow memory).

 Lower bounds on na\"ive matrix multiplication often use the Loomis-Whitney theorem, which embeds operations in the voxels of a computation cube~\cite{irony2004communication, ballard2013communication}. However, volume based arguments such as Loomis-Whitney do not apply for more general computations. I/O bounding techniques for algorithms beyond matrix multiplication generally focus on the computation graph itself. 
 
 Most current work on lower bounds via computation graphs requires manual inspection of the graph. In~\cite{redbluepebble}, the authors find a $2S$ partition of the computation graph to bound I/O---a proof technique that is non-trivial for complex graphs. In~\cite{pathrouting} and \cite{dichotomy}, the authors use path routing and dichotomy width respectively to find lower bounds. In~\cite{ballard2012graph}, the authors reduce the I/O problem to a graph partitioning problem in order to find a lower bound for Strassen's matrix multiplication algorithm using the edge expansion of the graph, which was computed by hand by recursively decomposing the Strassen computation graph. None of these methods can easily be computed automatically for arbitrary graphs. Instead, lower bounds on each graph must be separately proved by inspecting the specific graph, and are thus difficult to generalize. We instead focus on methods that can \textit{automatically} compute lower bounds for any input graph, regardless of its structure.

 To our knowledge, there are only two works that discuss automated methods for lower bounds for arbitrary graphs. In the first work, the authors find automatic bounds by computing convex min s-t cuts on the sub-graphs~\cite{elangostcut}. With a runtime $O(n^5)$ for a graph with $n$ nodes, this method is significantly slower than our spectral method, which can be computed in $O(n^3)$. We compare against this method in Section~\ref{evaluation} and find that it yields looser bounds than our proposed spectral method. The second work uses an Integer Linear Program (ILP) to solve for the $2S$ partition of the computation graph~\cite{elangodissertation}. This method is computationally expensive because it necessitates an exact ILP solver and is thus combinatorial in difficulty. Since this ILP based method is intractable, we do not compare its performance against the spectral bound as the method cannot be performed for large graphs, instead limiting ourselves to methods that can be computed in polynomial time. 
 


\section{Computation Graphs and Memory Model}
\label{memorymodel}
A computation can be represented by an underlying directed computation graph $G$. Each operation, including the inputs and outputs, is represented by a vertex. An edge from $u$ to $v$ indicates that the operation $v$ was computed with $u$ as an operand. The graph is acyclic, with the inputs as sources and the outputs as sinks. For example, the inner product of two vectors with two elements each can be represented as a 7 vertex graph: 4 vertices for inputs, 2 vertices for the intermediate products, and a single vertex for the sum. (Figure \ref{dotgraph}). 

\begin{figure}[t]
    \centering
    \includegraphics[width=0.33\textwidth]{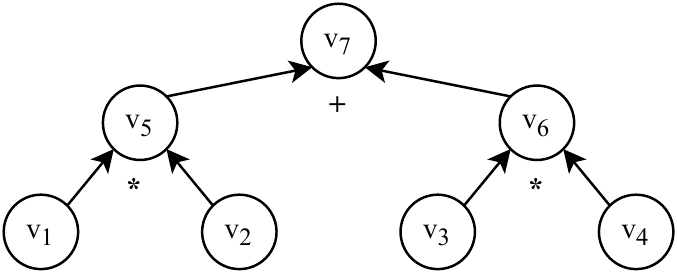}
    \caption{Computation graph of an inner product.}
    \vspace{-3mm}
    \label{dotgraph}
\end{figure}

We assume a two level memory hierarchy on a single processor with infinite slow memory and a limited cache of fast memory of size $M$ elements, where the result of each operation in the computation graph is a single element. Every operation in the computation graph must be evaluated. When a vertex $v$ is evaluated, the parents of $v$ must be loaded into fast memory from slow memory if they are not already present. As in \cite{ballard2012graph, pathrouting, elangodissertation}, we disallow recomputation of the same vertex: therefore, if a computed result is needed elsewhere in the computation graph and is about to be evicted, the result must first be written to slow memory. 

I/O can be separated into trivial (reading inputs and writing outputs) and non-trivial I/O. We focus on \textit{non-trivial} I/O: we thus do not directly include the cost of reading inputs or writing outputs. 
Instead, we assume that inputs can be read from the user directly into fast memory, and outputs are reported to the user immediately as they are computed. However: if an input is evicted from fast memory and is still needed elsewhere in the computation, it must be written to slow memory. This assumption is inherent in the proof in \cite{ballard2012graph, elangostcut}. Because we seek lower bounds, we do not constrain the eviction policy of fast memory. I/O is incurred when, during computation, an element is written to slow memory from fast memory or read from slow memory into fast memory. 

An evaluation order is then the order that operations are evaluated in the graph. Since a vertex can only be evaluated after its parents, a valid evaluation order must be topological with respect to the graph. We thus seek lower bounds on the I/O incurred by the optimal evaluation order.

\subsection{Optimization Task}
Formally, let $G = (V,E)$ be a computation graph with vertices $V$ and edges $E$. Let $n=|V|$ be the number of operations in the graph, and let $M$ be the size of fast memory. Note that each vertex in the graph is evaluated exactly once; therefore, the total computation takes exactly $n$ time-steps.

 We formalize an evaluation order on $G$ as a permutation matrix $X \in \mathbb{R}^{n \times n}$, where $X_{ij}$ is one if $v_j$ is computed at time-step $i$. Let $\mathcal{O}_G$ be the set of valid topological orders on $G$. Since vertices must be evaluated after their operands, $X \in \mathcal{O}_G$.

An I/O is incurred every time an element must be read into fast memory from slow memory or written to slow memory from fast memory. Let $J_G(X)$ be the number of nontrivial I/Os that were incurred by evaluating $G$ in the order specified by $X$ on $G$. We seek a lower bound on $J_G^*$, the optimal I/O incurred by any evaluation order:
$$J_G^* = \inf_{X \in \mathcal{O}_G} J_G(X).$$
\section{Spectral Bounds via the Graph Laplacian}
\newcommand{\floor}[1]{\ensuremath{\left \lfloor #1 \right \rfloor}}

\label{spectral}
 In this section, we find a lower bound based on the eigenvalues of the graph Laplacian. We first link the problem to the edge expansion of the graph, by counting the number of edges that cross boundaries over a graph partition as in \cite{ballard2012graph}. We frame this problem as a quadratic program (QP) with respect to the graph Laplacian. Finally we use the Laplacian's spectra to find a lower bound on the solution to the QP.

\textbf{Notation:} For $v \in V$, let $d_{in}(v), d_{out}(v),$ and $d(v)$ be the in-degree, out-degree, and total degree of $v$ respectively. Finally, for any subset $S \subseteq V$, we define $\partial S$ as the edge boundary of $S$:
$\partial S = \{ (u,v) \in E \mid (u \in S \wedge v \notin S) \vee (v \in S \wedge u \notin S) \}.$

\subsection{Counting Edges over Graph Partitions}

For any evaluation $X$ on $G$, we can choose a partition $P \subseteq 2^V$ that divides $V$ into disjoint subsets of vertices so each $S \in P$ is contiguously ordered by $X$. $P$ thus defines breakpoints on $X$. Figure \ref{segmentation} depicts an example of a partition on a graph. The numbers on the vertices indicate the evaluation order determined by $X$. The graph is then partitioned into green, yellow, and blue segments. Each segment is contiguous with respect to the order.

Let $\mathcal{P}_X$ be the set of valid partitions on $X$ according to the ordering constraint. We leverage the following key lemma from \cite{ballard2012graph}, which divides the I/O cost of a subset of a computation graph into reads (edges entering the subgraph), and writes (edges leaving the subgraph). For each subset $S \in P$, define the following sets:
$$R_S = \{v \in V \mid v \notin S, \exists (v, u)\in E \text{ s.t } u \in S\},$$
$$W_S = \{v \in V \mid v \in S, \exists (v,u) \in E \text{ s.t } u \notin S\}.$$

$R_S$ is the vertices not in $S$ with an edge into $S$, and $W_S$ is the vertices in $S$ with an edge outside of $S$. Ballard et. al in \cite{ballard2012graph} then present the following lemma:

\begin{lemma}[Equation 6 from \cite{ballard2012graph}]
$$J_G(X) \geq \max_{P \in \mathcal{P}_X} \left(\sum_{S \in P} |R_S| + |W_S|\right) - 2M|P|.$$
\end{lemma}
\begin{proof}
We summarize the proof of their lemma here. To evaluate the nodes in $S$, the vertices in $R_S$ must be read into fast memory (or were already in fast memory before beginning computation of $S$). Similarly, the vertices in $W_S$ are freshly computed and needed elsewhere in the evaluation and thus must be written out or left in fast memory at the end of $S$. (Figure \ref{partition}). Since the fast memory size is only $M$, at least $|R_S| + |W_S| - 2M$ I/O's are incurred by evaluating the nodes in $S$.
\begin{figure}[t]
    \includegraphics[width=0.3\textwidth]{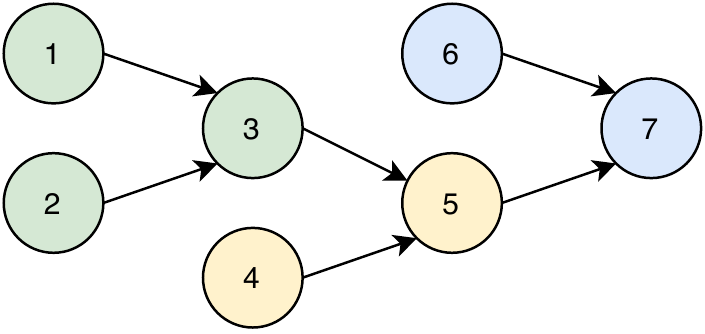}
        \caption{A computation graph: the numbers indicate the evaluation order and the colors are a valid partition.}
        \vspace{-2mm}
        \label{segmentation}
\end{figure}

\begin{figure}[t]
        \includegraphics[width=0.45\textwidth]{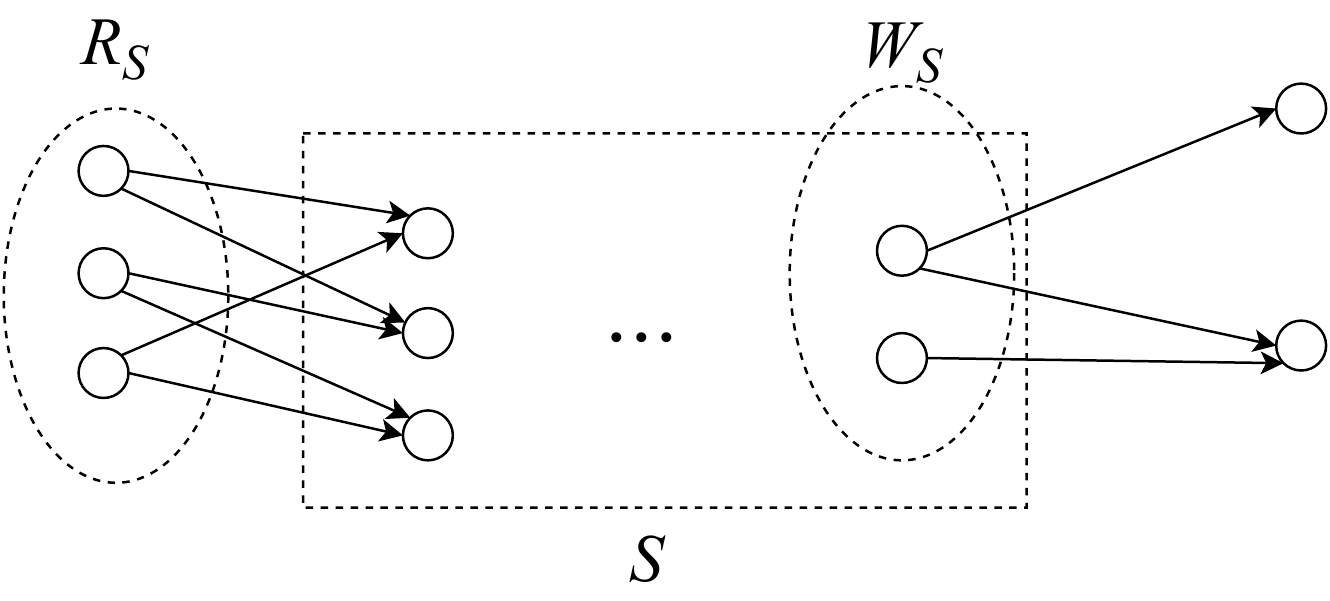}
        \caption{We identify sets $R_S$ and $W_S$ that cause I/O for each component $S$ in partition $P$. }
        \vspace{-3mm}
        \label{partition}    
\end{figure}

Summing over all $S \in P$ leads to a bound on the IO incurred by $G$. Any $P$ is valid so long as $P$ splits $V$ into components contiguous in $X$. Specifically, if $\mathcal{P}_X$ is the set of valid partitions with respect to $X$:
$$J_G(X) \geq \max_{P \in \mathcal{P}_X} \left(\sum_{S \in P} |R_S| + |W_S|\right) - 2M|P|.$$
\end{proof}

It is easier to compute the number of edges crossing into and out of $S$ rather than the vertex sets $R_S$ and $W_S$. Ballard et. al use this lemma to bound the I/O of the computation graph for Strassen matrix multiplication. However, they make several assumptions that weaken the bound for general graphs. Firstly, rather than computing a bound for all segments in the partition, they derive a bound for any single $n/|P|$ sized sub-graph within the Strassen computation graph. They then compute this bound specifically for Strassen-like graphs, and multiply this bound by $|P|$ to achieve a bound for the entire graph. This approach succeeds for the Strassen graph where the I/O is evenly distributed across the graph. However, this relaxation can be loose for graphs where the I/O is concentrated in a small portion of the vertices. Secondly, they deal strictly with regular graphs by adding loops to the computation graph. As a result, they link the size of $|R_S| + |W_S|$ to the size of the edge boundary by dividing by the maximum undirected degree, i.e
$|R_S| + |W_S| \geq \frac{1}{d_{max}(u)} |\partial S|.$ While this assumption is convenient for closed form bounds, it is not necessary for automatic methods where we can retain access to the graph. 

The following theorem links the partition to the I/O cost. We diverge from \cite{ballard2012graph} by bounding over all segments and maintaining access to the individual degrees of the vertices.

\begin{theorem} \label{objective}
    For fast memory size $M$ and graph $G$, the optimal I/O is lower bounded by:
    \begin{equation}\label{eq_objective}
        J^*_G \geq \min_{X \in \mathcal{O{_G}}}  \max_{P \in \mathcal{P}_X} \left(\sum_{S \in P}\sum_{(u,v) \in \partial S} \frac{1}{d_{out}(u)}\right) - 2M |P|.
    \end{equation}
\end{theorem}
\begin{proof}

 We bound $|R_S|$ and $|W_S|$ as:
$$|R_S| \geq \sum_{(u,v) \in E} \frac{\1\{u \notin S, v \in S\}}{d_{out}(u)}, \quad |W_S| \geq \sum_{(u,v) \in E} \frac{\1\{u \in S, v \notin S\}}{d_{out}(u)}.$$

Summing reads and writes, we have:
$$|R_S| + |W_S| \geq \sum_{(u,v) \in \partial S} \frac{1}{d_{out}(u)}.$$
Minimizing over all $X$, we get the full bound
\begin{equation}
    J^*_G \geq \min_{X \in \mathcal{O{_G}}}  \max_{P \in \mathcal{P}_X} \left(\sum_{S \in P}\sum_{(u,v) \in \partial S} \frac{1}{d_{out}(u)}\right) - 2M |P|.
\end{equation}
\end{proof}

Intuitively, an adversary picks some evaluation order $X$ on $G$. We pick a hard partition $P$ on $X$ to maximize the I/O incurred.  In the next section, we formalize Theorem \ref{objective} as a quadratic program using the graph Laplacian of an out-degree normalized graph. We then lower bound the I/O cost via the eigenvalues of the Laplacian. 

\subsection{Formulation via the Graph Laplacian}

In Theorem \ref{objective}, we solved for the minimum order over a maximum partition. However, since any partition will give us a lower bound, we can choose to split our graph into evenly sized segments. We pick some $k\leq n$ as our number of segments: splitting into $k$ subsets of as equally as possible (such that the first $n \mod k$ segments have $\floor{n/k} +1$ vertices and the rest have $\floor{n/k}$ vertices). For an evaluation order $X$, let $P^{(X,k)} \in P_\mathcal{X}$ be the $k$-partition described above. If $P^{(I,k)}$ would be the above partition assuming an identity evaluation order $X=I_k$, then we can define $\hat{W}^{(k)} \in \mathbb{R}^{n \times k}$ as $(\hat{W}^{(k)})_{ij} = \1\{i \in P^{(I,k)}_j\}$. Then $X\hat{W}^{(k)} \in \mathbb{R}^{n \times k}$ is the partition matrix for the $k$-partition $P^{(X,k)}$.

We transform our graph directed $G$ into a weighted undirected graph as follows: for each directed edge $(u,v) \in G$, we add the undirected edge $(u,v)$ to $\tilde{G}$ with weight $\frac{1}{d_{out}(u)}$. Henceforth, we indicate the degree function, degree matrix, and adjacency matrix of the original $G$ as $d(v), D$, and $A$ respectively; we analogously denote $\tilde{d}$, $\tilde{D}$, $\tilde{A}$ as the degree function, degree matrix, and adjacency matrix of $\tilde{G}$. 

Let $\tilde{L} = \tilde{D} - \tilde{A}$ be the graph Laplacian of $\tilde{G}$. $\tilde{L}$ is positive semi-definite, so all of its eigenvalues are nonnegative. The Laplacian is convenient for expressing the edge boundaries of vertex subsets. Specifically, for subset $S \subseteq V$, let $x \in \mathbb{R}^n$ be the one-hot encoding of $S$ (i.e $x_i = \1\{v_i \in S\}$). Then:
\begin{equation}\label{laplacianedges}
    x^T\tilde{L} x = x^T\tilde{D}x - x^T\tilde{A} x  = \sum_{(u,v) \in \partial S} \frac{1}{d_{out}(u)}.
\end{equation}

Using this property we can bound the edge crossing as:
$$\tr((\hat{W}^{(k)})^TX^T \tilde{L}X\hat{W}^{(k)}) = \sum_{S \in P^{(X,k)}} \sum_{(u,v) \in \partial S} \frac{1}{d_{out}(u)}.$$

Letting $W^{(k)} = \hat{W}^{(k)}\hat{W}^{(k)^T}$, and rewriting Equation \ref{eq_objective} leads to the following quadratic program:
\begin{theorem}[I/O Bound via Graph Laplacian]\label{bound_program}
    For a computation graph $G$ and any $k \leq n$ with $\tilde{L}$ and $W^{(k)}$ defined as above, $J_G^*$ is lower bounded by the solution of:
\begin{align*}
    \text{minimize}_X \quad & \max_{k} \tr(X^T\tilde{L}XW^{(k)}) - 2 k M\\
    & X \in \mathcal{O}_G.
\end{align*}
\label{laptheorem}
\end{theorem}
In the next section, we relax the above optimization problem to find a lower bound on the objective using the eigenvalues of $\tilde{L}$ and $W^{(k)}$.

\subsection{Spectral Bounds}

We derive the following eigenvalue bound:
\begin{theorem}[Spectral Method]\label{laplacianbound}
    \begin{equation} J_G^* \geq  \floor{\frac{n}{k}} \sum_{i=1}^k  \lambda_i(\tilde{L})  - 2kM.
    \end{equation}
\end{theorem}
\begin{proof}
We relax the topological constraint $X \in \mathcal{O}_G$, and instead constrain over orthogonal $X$. We thus have for any $k$:
$$J_G^* \geq \tr(X^T\tilde{L}XW^{(k)}) - 2kM \quad \text{s.t } X^TX=XX^T = I.$$

For symmetric $\tilde{L}, W$ and orthogonal matrix $X$, where $\lambda_1, ..., \lambda_n$ and $\mu_1, ..., \mu_n$ are the eigenvalues in increasing order of $\tilde{L}$ and $W$ respectively, we have $\tr(X^T\tilde{L}XW) \geq \sum_{i=1}^n \lambda_i \mu_{n-i}$, or the minimal dot product of $\lambda$ and $\mu$ (see \cite{finke1987quadratic}, Theorem 3). Here  $W^{(k)}$ is a block diagonal matrix, with $n-k$ zero eigenvalues and  $k$ eigenvalues that are at least $\floor{n/k}$. Therefore, we apply our lower bound as a sum of the first $k$ eigenvalues of $\tilde{L}$:
\begin{equation*} J_G^*  \geq \max_k \tr(X^T\tilde{L}XW^{(k)}) - 2kM \geq \sum_{i=1}^k \floor{n/k} \lambda_i(\tilde{L})  - 2kM,
\end{equation*}
\end{proof}

This bound can be found in $O(n^3)$ time. We first find the eigenvalues $\lambda(\tilde{L})$ in $O(n^3)$. We then iterate over possible values of $k$ which takes constant time per iteration to find the best eigenvalue. However, we generally only need small number of eigenvalues to find a good $k$. Since any value of $k$ is a lower bound, it suffices to find the $h$ smallest eigenvalues of $L$. These values can be found using a method such as Lanczos-Arnoldi with time complexity $O(hn^2)$: this complexity decreases even further with sparse $L$ using sparse eigenvalue solvers. 

For closed form analysis, sometimes the exact form of the original Laplacian spectra $\lambda(L)$ are known, but the spectra of our out-degree normalized Laplacian $\lambda(\tilde{L})$ are not. While $\lambda(\tilde{L})$ can be easily computed automatically, they can be harder to derive for closed form analysis. We can naturally loosen the bound in Theorem \ref{laplacianbound} to be in terms of $L$ rather than $\tilde{L}$. 
\begin{theorem}[Spectral Method with Original Graph Laplacian]\label{looserlaplacianbound}
    \begin{equation} J_G^* \geq  \frac{1}{\max_{v \in V} d_{out}(v)} \floor{\frac{n}{k}} \sum_{i=1}^k  \lambda_i(L)  - 2kM.
    \end{equation}
\end{theorem}
\begin{proof}
We follow the same steps of Theorem \ref{laplacianbound}, but we bound Equation \ref{eq_objective} as 
$$J^*_G \geq \min_{X \in \mathcal{O{_G}}}  \max_{P \in \mathcal{P}_X} \left(\sum_{S \in P} \frac{|\partial S|}{\max_{v \in V} d_{out}(v)} \right) - 2M |P|.
$$

We can then reframe Equation \ref{laplacianedges} in terms of $L$, noting that, if $x \in \mathbb{R}^n$ is the one-hot encoding of $S \subseteq V$, then $x^TLx = |\partial S|.$ Using the same partitioning argument and definition of $W^{(k)}$, we can reframe the quadratic program in Theorem \ref{bound_program} with $L$ instead of $\tilde{L}$:
\begin{align*}
    \text{minimize}_X \quad & \max_{k} \frac{\tr(X^TLXW^{(k)})}{\max_{v \in V} d_{out}(v)}  - 2 k M\\
    & X \in \mathcal{O}_G.
\end{align*}

Then, following the same spectral argument as in Theorem \ref{laplacianbound}, we arrive at a looser, but more convenient bound:
$$J_G^* \geq  \frac{1}{\max_{v \in V} d_{out}(v)} \floor{\frac{n}{k}} \sum_{i=1}^k  \lambda_i(L)  - 2kM.$$
\end{proof}

\subsection{Parallel Spectral Bounds}
We generalize Theorem \ref{laplacianbound} to the parallel setting as follows. Suppose that we have $p$ processors, each with memory $M$. As in \cite{irony2004communication,ballard2012graph} we count I/O as the communication with a processor to slow memory or between processors. We make no assumptions about the distribution of the workload.

\begin{theorem}[Parallel Spectral Bound]
For a computation graph $G$ distributed across $p$ processors, at least one of the processors has I/O $J_G^*$ lower bounded by:
\begin{equation} J_G^* \geq  \floor{\frac{n}{kp}}  \sum_{i=1}^k \lambda_i(\tilde{L})  - 2kM.
\end{equation}    
\end{theorem}
\begin{proof}
    For a given evaluation of $G$, each vertex in $G$ is evaluated by one processor. Let $V_1,...,V_p$ be the vertex sets associated with each processor. Then given the optimal evaluation order $X$ we can define $p$ evaluation orders $X_1,...,X_p$ where $X_i$ indicates the evaluation order of processor $i$ on its vertex set $V_i$. Since memory is local to each processor, we can use the same graph partitioning machinery in Theorem \ref{laplacianbound} per-processor. For processor $i$, if $\mathcal{P}_{X_i}$ is the set of valid partitions over $X_i$, the I/O of processor $i$ is lower bounded by
    $$J_G(X_i) \geq \max_{P \in \mathcal{P}_{X_i}} \left(\sum_{S \in P} |R_S| + |W_S|\right) - 2M|P|.$$

    There must exist one processor $i^*$ for which $|V_{i^*}| \geq n/p$. For this processor, we can partition $V_{i^*}$ into $k$ equal parts of $\frac{n}{kp}$ vertices each (call this partition $P$). Since any partition of $V_{i^*}$ is a lower bound, we have 
    $$J_G(X_{i^*}) \geq \left(\sum_{S \in P}\sum_{(u,v) \in \partial S} \frac{1}{d_{out}(u)}\right) - 2kM.$$

    However, as a looser lower bound, instead of restricting $P$ to equal partitions of $V_{i^*}$, we can consider the set of equal partitions of the entire graph into $kp$ parts of $\frac{n}{kp}$, and simply pick the $k$ sections incurring the least I/O cost. This is then equivalent to the bound in Theorem \ref{laplacianbound} with $kp$ partitions, but we only count I/O from the first $k$ parts (which correspond to the smallest $k$ eigenvalues). We then have 
    $$J_G^* \geq   \floor{\frac{n}{kp}} \sum_{i=1}^k \lambda_i(\tilde{L})  - 2kM.$$
\end{proof}
\section{Analytical Bounds for Specific Graphs}
For computation graphs with known eigenvalues, we can compute the bound in Theorem \ref{looserlaplacianbound} directly. 
We perform this analysis for the Bellman-Held-Karp algorithm for the Traveling Salesman Problem and as well as the Fast Fourier Transform, which have a hypercube and butterfly computation graph respectively. For both of these problems, we consider solely \textit{nontrivial I/O}, which does not count reading inputs or writing outputs. In the process, we derive the spectrum of the FFT graph in Appendix \ref{ffteigappendix}; to our knowledge, this is the first closed form of the spectrum of the unwrapped butterfly graph that includes multiplicity.  Finally, we present a probabilistic bound on the I/O of a random Erd\H{o}s R\'enyi graph.

Previously, \cite{redbluepebble} found an asymptotically tight bound of $\Omega(\frac{l 2^l}{\log M})$ for a $2^l$ point FFT through manual inspection of $2S$ partitions. We find that our spectral bound of $\Omega(\frac{l 2^l}{\log^2 M})$ is only a factor of $1/\log M$ off from this published tight bound.

\subsection{Hypercube Graph}
\label{hypercube_closedform}
\begin{figure}[t]
    \includegraphics[width=0.3\textwidth]{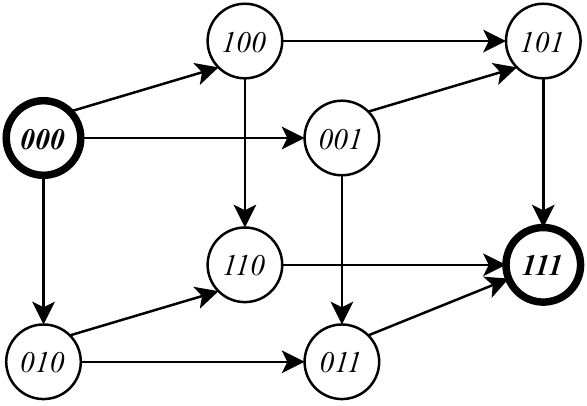}
    \caption{Bellman-Held-Karp computation graph with 3 cities. The graph is a 3 dimensional boolean hypercube, with starting point $000$ and ending point $111$.}
    \vspace{-5mm}
    \label{theoreticalgraphs_hypercube}
\end{figure}

The hypercube is a computation graph that appears as a result of many hard dynamic programming problems \cite{ambainis2019quantum}. For example, consider the well-known Bellman-Held-Karp algorithm which uses dynamic programming to solve the traveling salesman problem \cite{bellman1962dynamic,held1962dynamic}. The approach solves for the optimal path visiting a subset of the cities by leveraging the computed optimal paths through adjacent subsets with one fewer city. 

The Bellman-Held-Karp algorithm with $l$ cities can be naturally formulated as an iteration on the vertices of a boolean hypercube. We encode ``cities visited'' as a length $l$ binary string. Let $Q_l$ be a boolean $l$-dimensional hypercube, where each vertex $k$ is a length $l$ binary string, and $(k_1, k_2) \in E$ if $k_2$ can be constructed by setting a 0 in $k_1$ to 1. Let $S(k, i)$ be the shortest path visiting all the cities active in $k$ and ending up at the $i$'th city. Then to solve TSP, we seek to find the solution set $Y[k] = \{S(k, i) \mid k[i] = 1\}$ for each vertex $k$ in the boolean hypercube $Q_l$. For example, $Y[01101]$ contains three paths that have traversed cities 2, 3, and 5, where each path has a different ending point. $Y[k]$ can be easily computed given the results of $k$'s incoming neighbors in $Q_l$, and $Y[\{1\}^l]$ gives the solution to the TSP. Thus $Q_l$ represents the computation graph for the under this formulation of the Bellman-Held-Karp algorithm. An example of this graph can be seen in Figure \ref{theoreticalgraphs_hypercube}.

The I/O bound for this formulation of the Bellman-Held-Karp algorithm can then be found via our spectral method, because the hypercube has relatively simple eigenvalues. The $l$-dimensional hypercube has $n=2^l$ vertices and Laplacian eigenvalues $2i$ for $i=0,...,l$ with multiplicity ${l \choose i}$. If we choose $k$ to encompass the top eigenvalues up to $i=\alpha$, we have $k = \sum_{i=0}^\alpha {l \choose i}$. The maximal out degree is $l$. For any $\alpha < 2^l$:
\begin{align*} J_G^*  &\geq \frac{1}{\hat{d}_{out}} \frac{2^{l+1}}{\sum_{i=0}^\alpha {l \choose i}} \sum_{i=0}^\alpha i {l \choose i} - 2M \sum_{i=0}^\alpha {l \choose i}\\
&=  \sum_{i=0}^\alpha {l \choose i} \left( 
    i \frac{2^{l+1}}{l \sum_{i=0}^\alpha {l \choose i}} - 2M
    \right).
\end{align*}

While any $\alpha<l$ would be a lower bound, for simplicity we here choose $\alpha = 1$ (i.e $k=l+1$):
$$J_G^*  \geq  \frac{2^{l+1}}{(l+1)} - 2M(l+1) .$$
For a tighter bound we can optimize more specifically over $\alpha$. We see that this bound is nontrivial as long as $M \leq \frac{2^{l}}{(l+1)^2}$.



\subsection{Fast Fourier Transform Graph}

\begin{figure}[t]
    \includegraphics[width=0.35\textwidth]{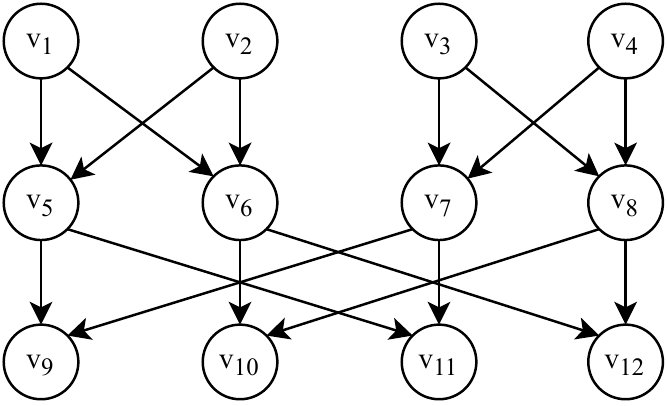}
    \caption{Computation graph for a $2^l=4$ point Fast Fourier Transform.}
    \label{theoreticalgraphs_fft}
    \vspace{-5mm}
\end{figure}

The Fast Fourier Transform (FFT) computation graph is a butterfly graph. For a $2^l$ point FFT, the butterfly graph $B_l$ has $(l+1)2^l$ vertices, which can be arranged into $l+1$ columns. A $2^l=4$ point FFT can be seen in Figure \ref{theoreticalgraphs_fft}.  The butterfly graph can be defined inductively: allow $B_0$ to be defined as a single vertex. Then $B_l$ can be constructed as two copies of $B_{l-1}$ that are joined by an extra final column of $(l+1)$ vertices.

We derive the eigenvalues of the Laplacian of $B_l$ in Appendix \ref{ffteigappendix}. To our knowledge, closed forms with multiplicities were only previously known for the wrapped butterfly graph \cite{comellas2003spectra}. The eigenvalues we derive are:
\begin{align*}
    &4-4 \cos\left(\frac{\pi j}{l+1}\right), \forall j=0,...,l; \; \text{multiplicity } 1\\
    &4-4\cos \left(\frac{\pi(2j+1)}{2i+1}\right)\forall i=1,...,l; j=0,...,i-1; \; \text{multiplicity } 2^{l-i+1}\\
    & 4 - 4\cos\left( \frac{j\pi}{i+1}\right) \forall i=1,...,l-1; j=1,...,i; \text{multiplicity } (l-i)2^{l-i-1}.
\end{align*}

The smallest eigenvalue is 0 (from the first expression), but the next eigenvalues are governed by the second expression with $j=0$ as long as $2i+1 \geq l+1$. We choose some $\alpha$ and set $k = 2^{\alpha+1} $. We compute the lowest $k$ eigenvalues of the Laplacian of $B_l$. Of these eigenvalues, $2^{\alpha}$ have (with $i=l-\alpha$):
$$\lambda = 4-4\cos \left(\frac{\pi}{2(l-\alpha)+1}\right).$$
To compute our lower bound, we assume the other eigenvalues are 0. We note that $n = (l+1)2^l$. Then we have (dividing by our maximal out-degree $2$):
\begin{align*}
    J_G^* \geq  (l+1)2^{l}\left(1 - \cos \left(\frac{\pi}{2(l-\alpha)+1}\right) \right) - 2^{\alpha+2}M.
\end{align*}
Suppose that we set $\alpha = l - \log_2 M$, under the assumption that $M \ll l$. Then:
$$J_G^* \geq  (l+1)2^{l}\left (
    1 - \cos \left(\frac{\pi}{2 \log_2 M+1}\right) - \frac{4}{l+1}
\right).$$
To see how this behaves, we can use the small angle approximation $\theta^2/2 \approx 1- \cos(\theta)$ for small $\theta$ to get:
$$J_G^* \geq  (l+1)2^{l}\left(\frac{\pi^2}{ 8\log_2^2 M } - \frac{4}{l+1}\right).$$

Thus, for large $M$ and $l$ where $M \ll l$, our bound behaves at least as well as $\Omega (\frac{l2^l}{\log^2 M})$. This bound is only a $1/\log_2 M$  factor worse than the tight lower bound for butterfly graphs: $\Omega (\frac{l 2^l}{\log M})$, which is computed by inspection on the specific graph using S-partitions \cite{redbluepebble}.  

\subsection{Random Graphs}
The spectral bound is flexible, and can perform well on most graphs with high connectivity regardless of its structure. In the following section, we characterize the performance of this spectral bound given a random graph, and show that the spectral bound provides nontrivial results for as long as the graph is well-connected. While this graph is not a specific computation graph, examining random graphs allows us to understand the performance of the bound as we increase the connectivity of the graph. 

We consider an Erd\H{o}s R\'enyi graph $G(n,p)$ on $n$ vertices, where each edge is determined by flipping a coin with probability $p$. We will only deal with the regime where $p \geq \frac{\log n}{n}$, where the graph is almost surely connected \cite{kolokolnikov2014algebraic}.

We begin with the case where $p = \Theta(\frac{\log n}{n})$, but the graph is still connected. More specifically, we specify $p = p_0 \frac{\log n}{n-1}$ for some $p_0 > 6$. By \cite{kolokolnikov2014algebraic}, in this regime:
$$\lambda_2 \sim p_0 \log n\left( 1 - \sqrt{\frac{2}{p_0}} + O(\frac{1}{p_0}) + O(\frac{1}{\sqrt{p_0 \log n}})\right).$$

We first concentrate the maximum degree of the graph using Chernoff's bound as in \cite{chung2011spectra}. We first note that the degree $d$ of a single vertex is governed by the sum of $n-1$ Bernoulli random variables with probability $p$. The expected degree is $\mu = p (n-1) = p_0\log n$. Then using Chernoff's bound, we have:
\begin{align*}
    P(d \geq (1+\delta)\mu) &\leq \exp(-\mu \delta^2/3)\\
    P(d \geq (1+\delta) p_0 \log n) & \leq  \exp(\frac{-\delta^2p_0 \log n }{3}).
\end{align*}

If we set $\delta = \sqrt{6/p_0}$, we concentrate individual degrees as 
$P(d \geq (1+\sqrt{6/p_0})p_0 \log n) \leq 1/n^2.$
Then, using the union bound, we can concentrate the maximum degree as:
$$P(d_{max} \geq (1+\sqrt{6/p_0})p_0 \log n) \leq 1/n$$

Thus, with high probability ($1/n \rightarrow 0$ as $n\rightarrow \infty$), we have 
$$d_{max} \geq (1+\sqrt{6/p_0})p_0 \log n.$$

Setting $k=2$ in Theorem \ref{looserlaplacianbound}, we have with high probability:

$$J_G^* \leq \frac{n}{(1+\sqrt{6/p_0})}\left( 1 - \sqrt{\frac{2}{p_0}} + O(\frac{1}{p_0}) + O(\frac{1}{\sqrt{p_0 \log n}})\right) - 4M.$$

As $n \rightarrow \infty$ this bound scales roughly with $n + \frac{n}{\sqrt{\log n}}$, and is linear in $M$. Our bound becomes weaker, but still nontrivial, when we consider regimes with higher $p$. This is because as $p$ increases, the maximum degree scales to almost $np$ (and our bound requires dividing by the maximum out degree). For example, consider the case where $\frac{np}{\log n}\rightarrow \infty$, as in this regime the graph is essentially regular with degree $np$. Then from \cite{kolokolnikov2014algebraic}, we have that with high probability as $n \rightarrow \infty$:
$$\lambda_2(L) = np + O(\sqrt{np \log n}).$$
We then can apply Theorem \ref{looserlaplacianbound} to lower bound the non-trivial I/O (setting $k=2$) and dividing by the max degree $np$:
$$J_G^* \geq \frac{n}{2}(1 + O(\sqrt{\frac{\log n}{np}})) - 4M.$$
As $n \rightarrow \infty$, $O(\sqrt{\frac{\log n}{np}})$ will decay to zero resulting in a bound linear in $n$. 
\section{Evaluation}
\begin{figure}
    \minipage{0.35\textwidth}
        \includegraphics[width=\linewidth]{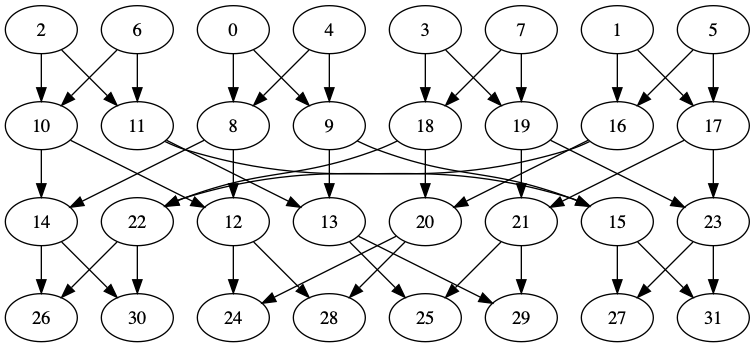}
    \endminipage
    \vspace{4pt}
    \\
    \minipage{0.35\textwidth}
    \includegraphics[width=\linewidth]{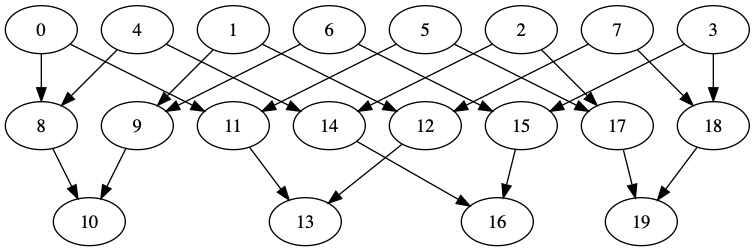}
    \endminipage
    \vspace{4pt}
    \\
    \minipage{0.45\textwidth}
    \includegraphics[width=\linewidth]{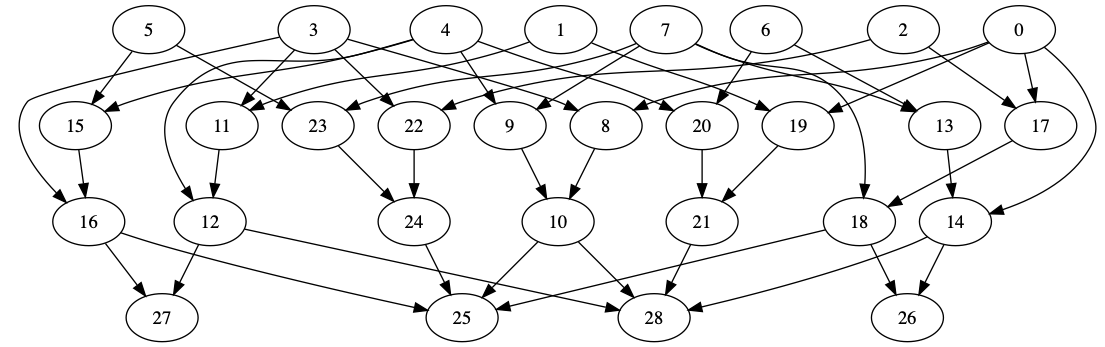}
    \endminipage
    \vspace{4pt}
    \\
    \minipage{0.4\textwidth}
    \includegraphics[width=\linewidth]{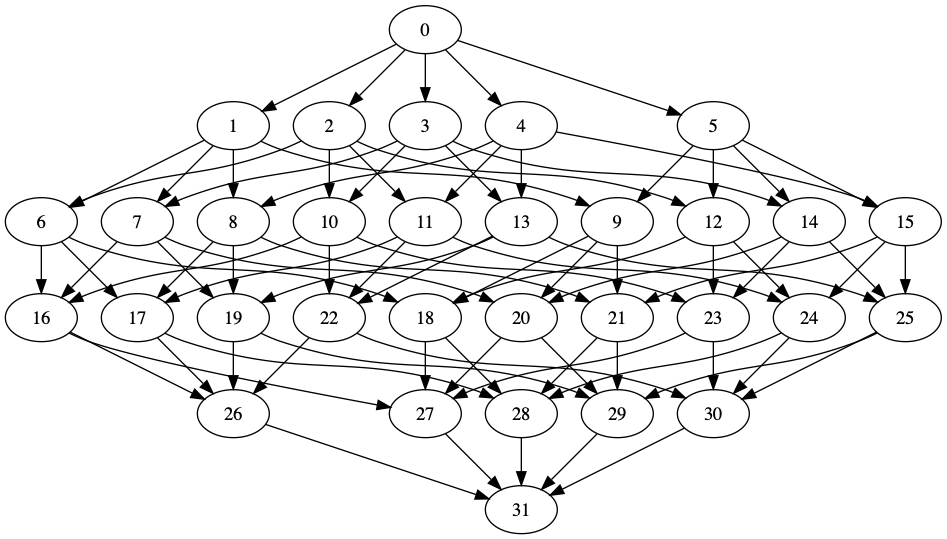}
    \endminipage
    \vspace{4pt}
    \\
    \caption{Computation Graphs for (top to bottom): 8 point FFT, $2 \times 2$ Naive Multiplication, $2 \times 2$ Strassen Multiplication, 5 city Bellman-Held-Karp algorithm for TSP}
    \label{graphexamples}
\end{figure}

\label{evaluation}
\subsection{Solver}
We evaluate our two lower bounds on four common computation graphs. To facilitate our evaluation, we develop a solver that traces operations during a Python computation and thus extracts a computation graph \footnote{Our code can be found at \url{https://github.com/stanford-futuredata/graphIO}}. The solver inter-operates with standard arithmetic operations and supports the inclusion of custom operations.

When computing Theorem \ref{laplacianbound}, any choice of $k < n$ produces a valid bound. We set $h=100$, computing up to the first 100 values of the graph Laplacian, and choose the optimal $k$ from $k \in \{2...h\}$. We discuss this choice in Section \ref{scalability}.

\subsection{Evaluation Computation Graphs}
\label{specification}
We evaluate the following graphs. Examples of these graphs can be found in Figure \ref{graphexamples}.

\begin{enumerate}
    \item \textit{Fast Fourier Transform (FFT):} We evaluate the $l$ level FFT of an $2^l$ element array, which is a butterfly graph. This graph has a published bound \cite{redbluepebble} of 
    $$\Omega \left(l 2^l/\log M\right).$$
    \item \textit{Naive Matrix Multiplication:} We evaluate the graph formed by naive multiplication of two $n\times n$ matrices $C = AB$. Specifically, we compute $C_{ij}$ as the dot product of the $i$th row of $A$ and the $j$th column of $B$. This graph has a published bound \cite{irony2004communication} 
    $$\Omega \left(n^3/\sqrt{M} \right).$$
    \item \textit{Strassen Multiplication} We evaluate the graph formed by multiplying to $n \times n$ matrices $C=AB$ via Strassen's method. Since Strassen's method is a recursive method that splits matrices into quadrants, we evaluate on values of $n$ that are powers of 2. This graph has a published bound \cite{ballard2012graph} of 
    $$\Omega \left(\left(n/\sqrt{M}\right)^{\log_2 7} M\right).$$
    \item \textit{Bellman-Held-Karp} We evaluate the hypercube computation graph formed by performing the Bellman-Held-Karp algorithm for a $l$ city TSP. We could not find a prior I/O bound for this problem in the current literature. However, in Section \ref{hypercube_closedform} we derive using the spectral method a bound of:
    $$\Omega\left(\left(2^l/l\right) - 2Ml\right).$$
\end{enumerate}

\begin{figure}[t]
    \includegraphics[width=0.5\textwidth]{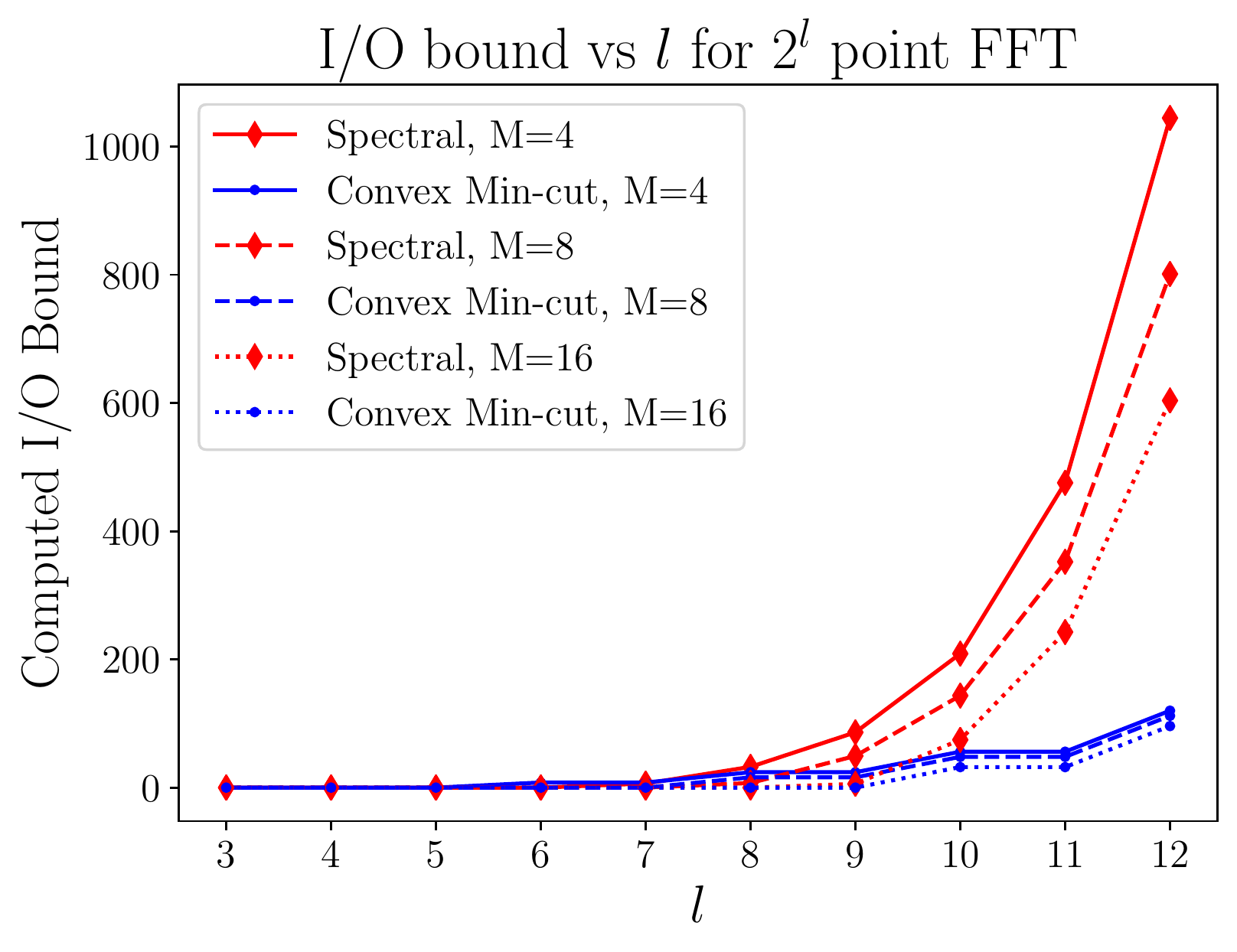}
    \\
    \includegraphics[width=0.5\textwidth]{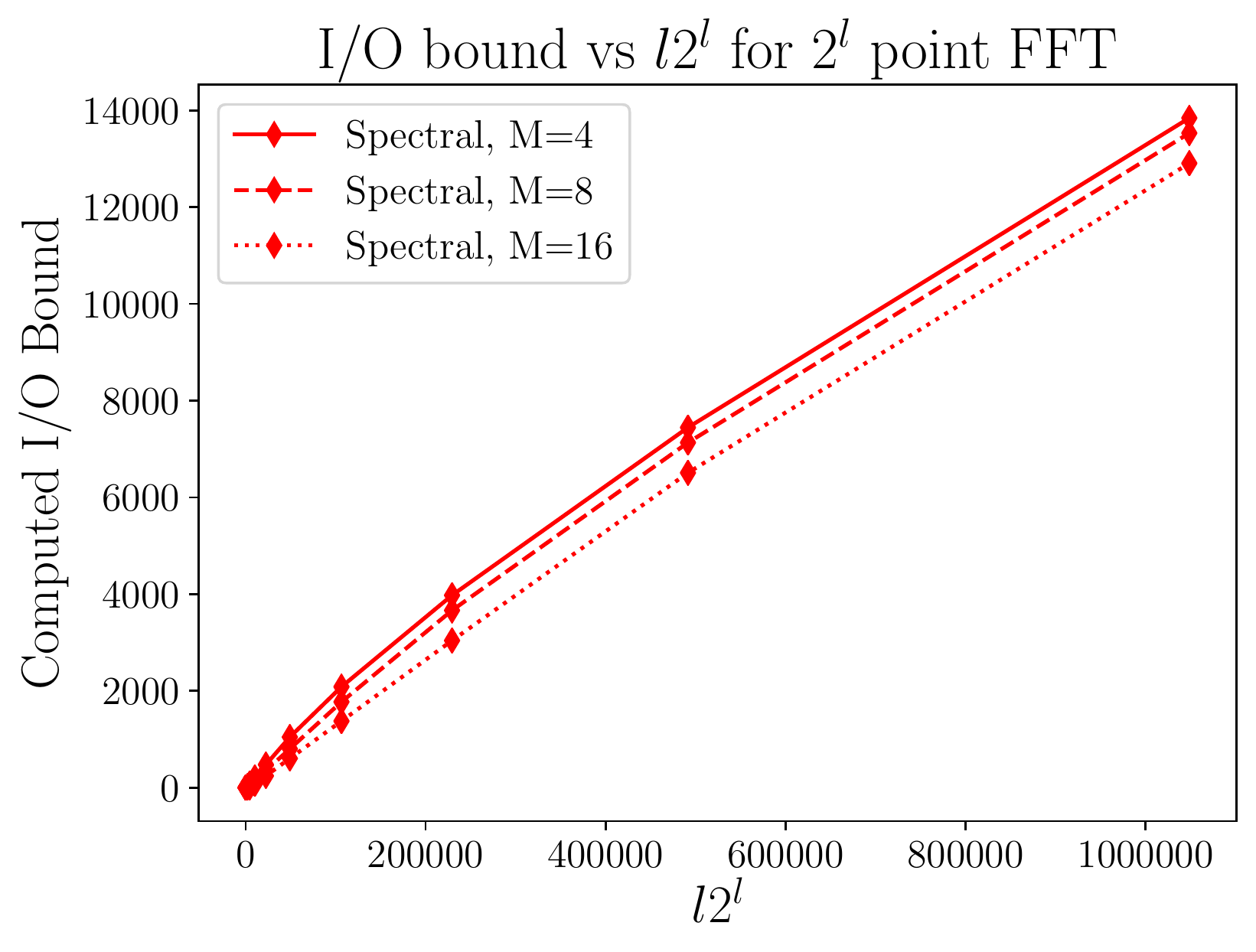}
    \caption{FFT: Bound vs $l$ (top) and $l 2^l$ (bottom) for $M=4, 8, 16$; $l$ = FFT Level. Max in-degree 2}
    \label{fftbounds}
\end{figure}

\begin{figure}[t]
    \includegraphics[width=0.5\textwidth]{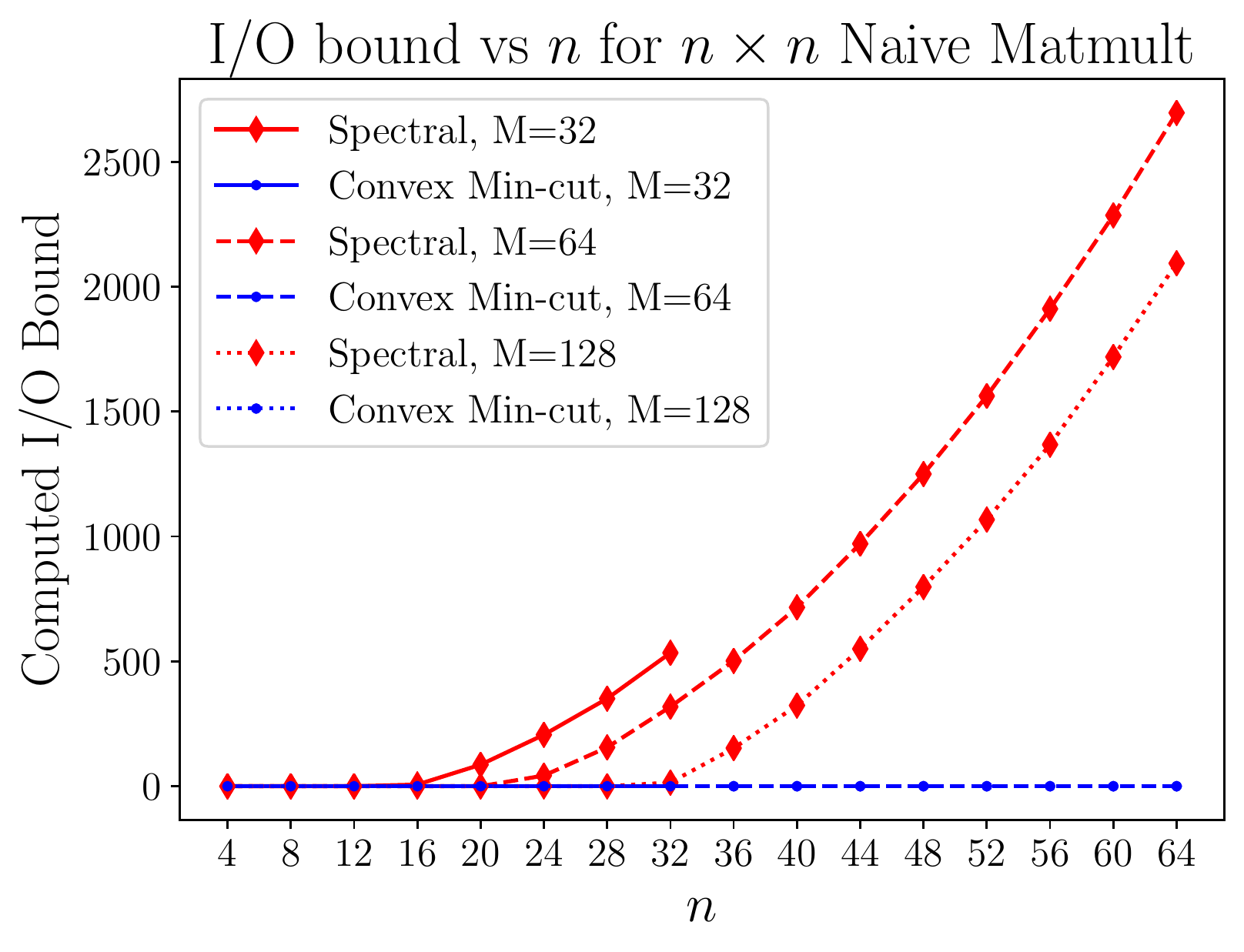}
    \hfill
    \includegraphics[width=0.5\textwidth]{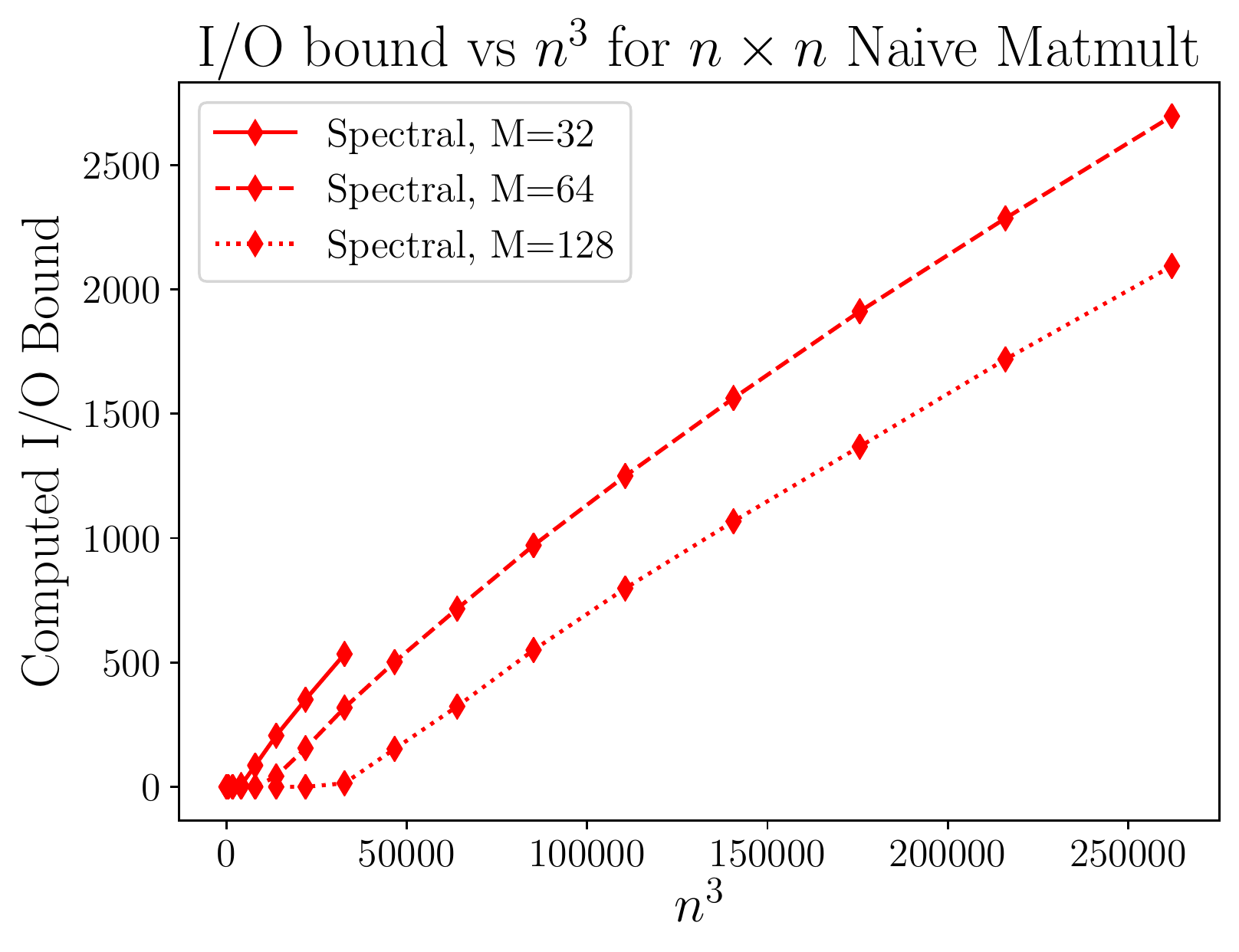}
    \caption{Naive matrix multiplication: Bound vs $n$ (top) and $n^3$ (bottom) for $M=[32, 64, 128]$; $n$ = side length. Max in-degree $n$.}
    \vspace{-5mm}
    \label{matmultbounds}
\end{figure}

\begin{figure}[t]
    \includegraphics[width=0.5\textwidth]{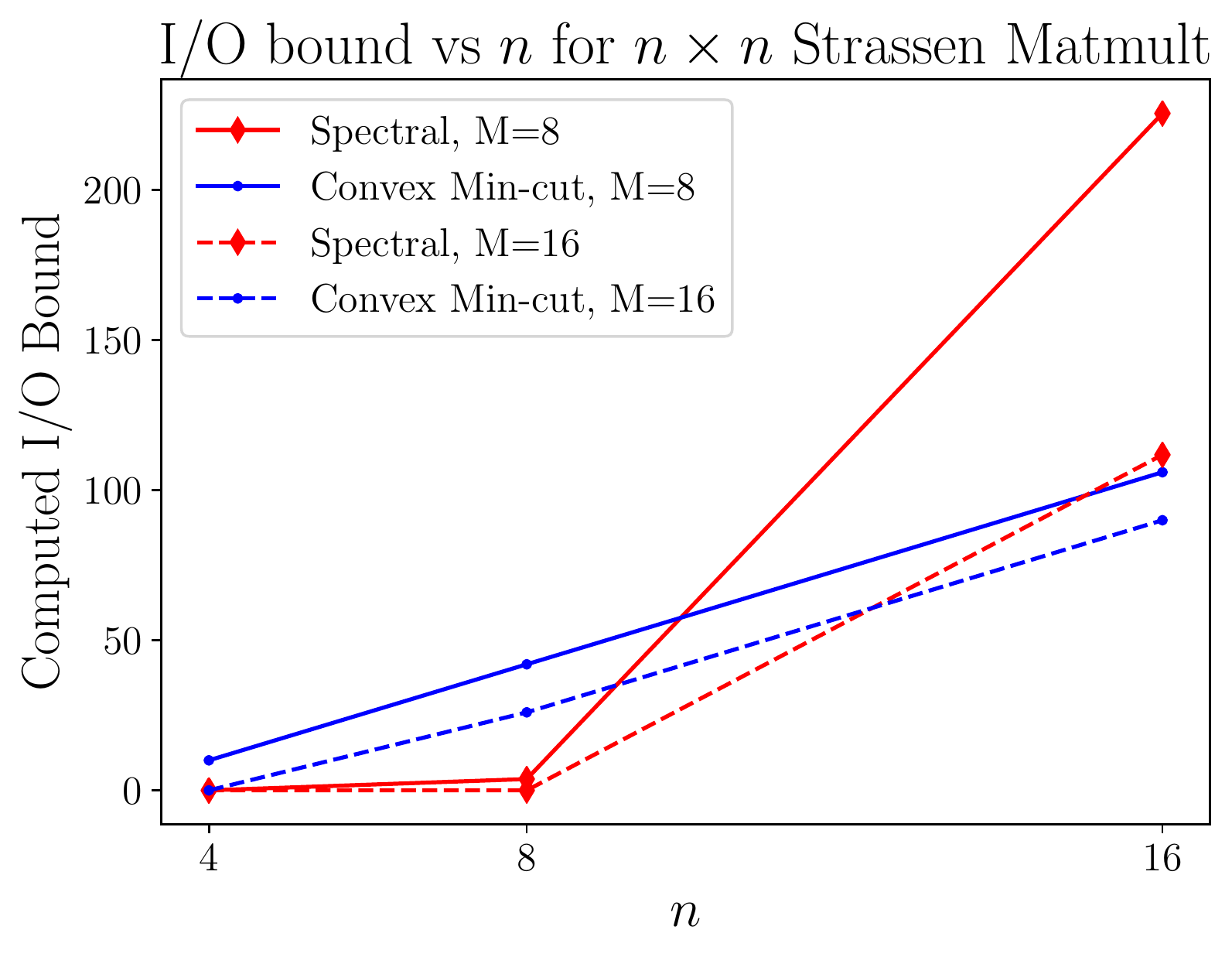}
    \hfill
    \includegraphics[width=0.5\textwidth]{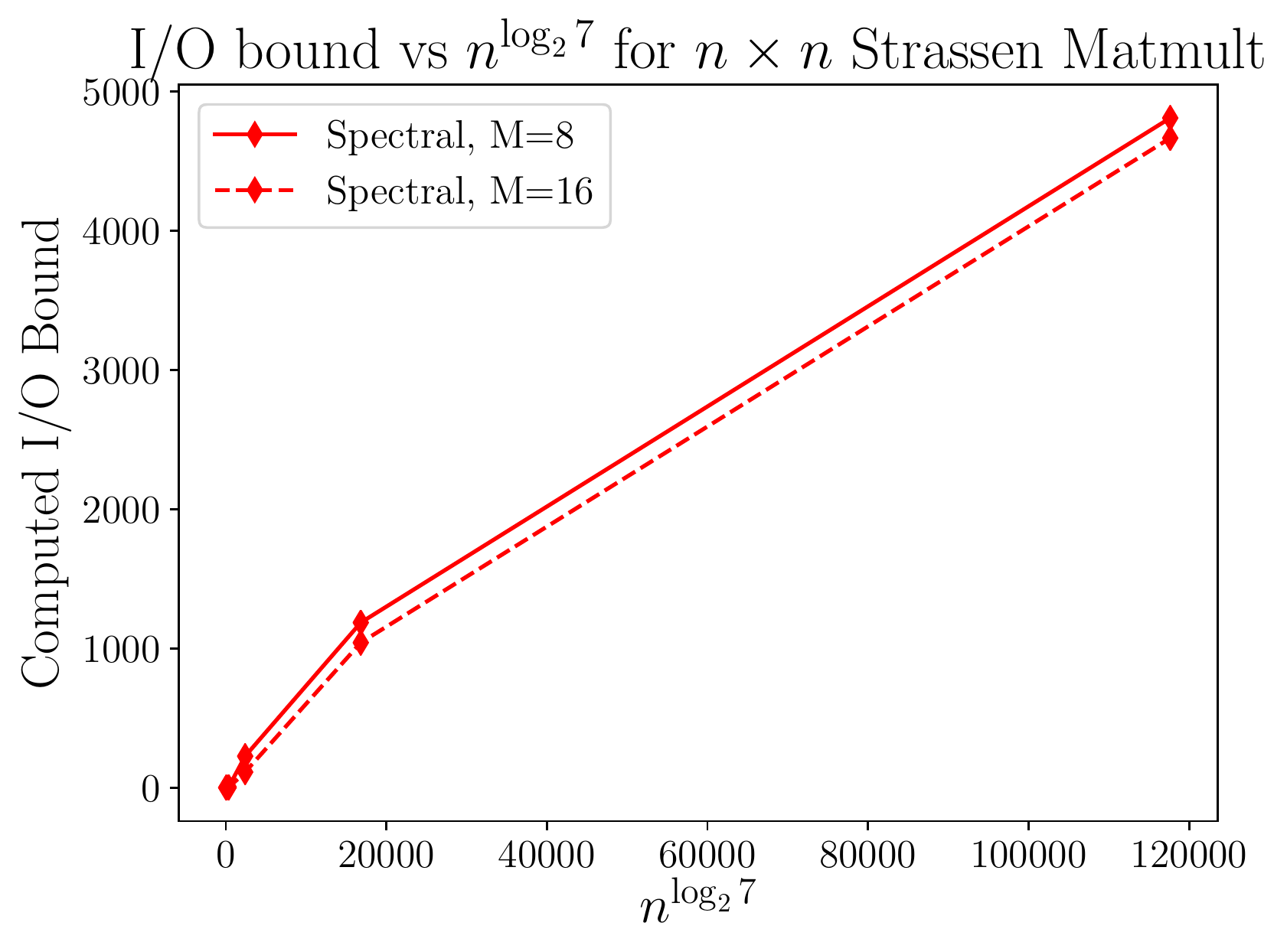}
    \caption{Strassen: Bound vs $n$ (top) and $n^{\log_2 7}$ (bottom) for $M=[8, 16]$; $n$ = side length. Max in-degree 4.}
    \label{strassenbounds}
\end{figure}

\begin{figure}[t]
    \includegraphics[width=0.5\textwidth]{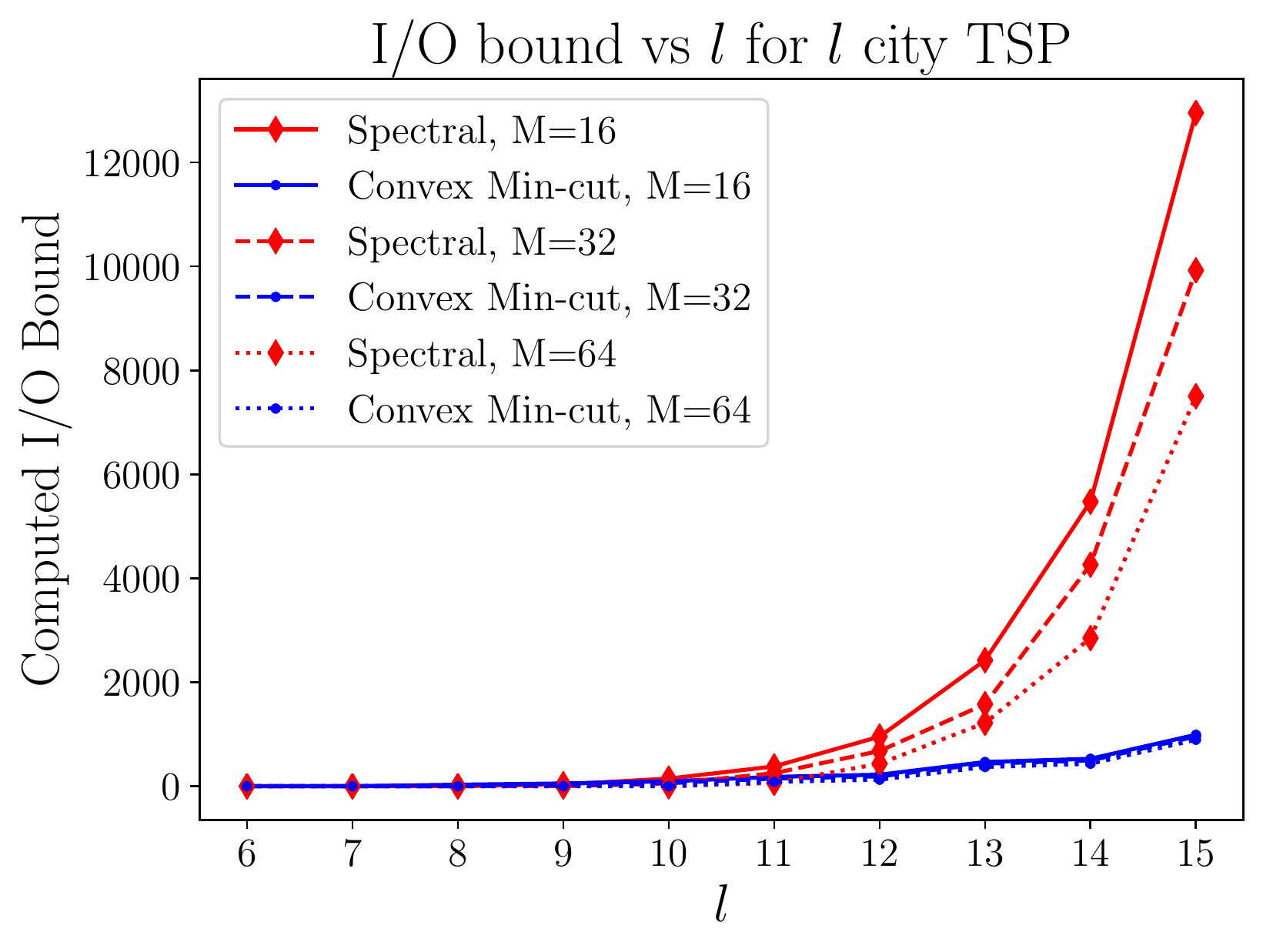}
    \\
    \includegraphics[width=0.5\textwidth]{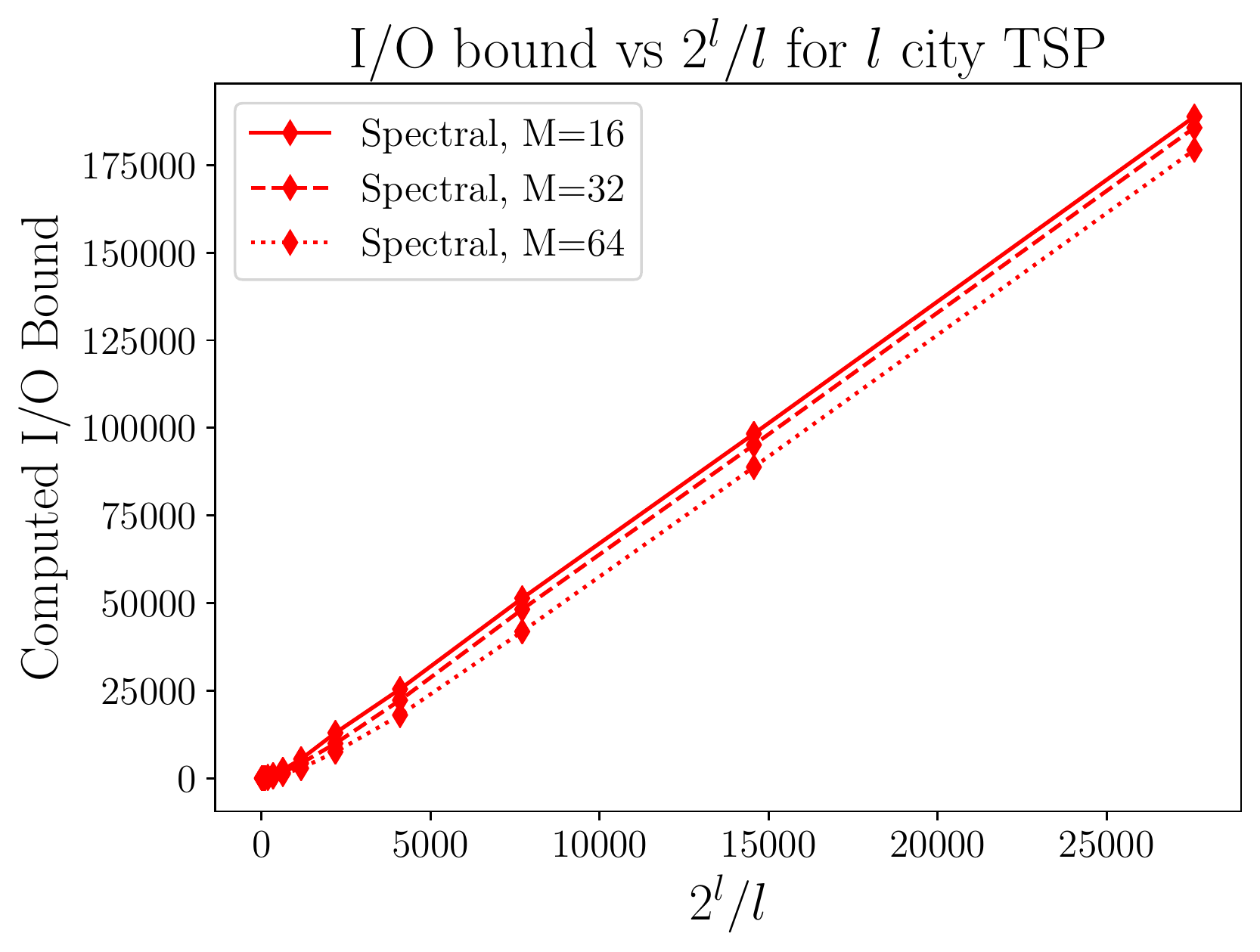}
    \caption{Bellman-Held-Karp for TSP: Bound vs $l$ (top) and $2^l/l$ (bottom) for $M=16, 32, 64$; $l$ = number of cities. Max in-degree $l$}
    \label{tspbounds}
\end{figure}

\subsection{Baselines} 
The only current methods for creating automatic lower bounds for any arbitrary graph that we could find are the convex min-cut method \cite{elangostcut} and the $2S$ partition method \cite{elangodissertation} (see Section \ref{relatedwork}). The $2S$ partition method involves solving a Mixed Integer Program, and is combinatorial in complexity: we could thus not perform this method for large graphs. The convex min-cut method is polynomial in time but still extremely expensive at $O(n^5)$. We evaluate the convex min-cut method for as large graphs as possible, cutting off evaluation at 1 day.

\textbf{Convex-Min Cut:} The convex min-cut method transforms the graph with respect to a vertex $v$ into a flow problem and then finds the minimum s-t cut of the transformed graph. The method maximizes over all $v$ in the graph. The method decomposes trivial (reading inputs and writing outputs) and non-trivial I/O, and thus fits well with our problem set-up. The runtime of this bound is $O(n^5)$ where $n$ is the number of vertices in the graph. In order to reduce runtime, the authors suggest loosening the bound by partitioning the graph into smaller sub-graphs using METIS and running convex min-cut on each sub-graph. If $C(v, G)$ is the minimum convex cut for $G$ transformed with respect to $v$, and $\mathcal{P}$ is the partition reported by METIS, the authors report the bound:
$$J_G^* \geq \sum_{P \in \mathcal{P}} \max_{v \in P} \max(0, 2*(C(v, G)-M)).$$

More details can be found in their paper. The authors suggest that each sub-graph have at most $2*M$ vertices, and evaluate their bounds on a series of small, simple computation graphs with very uniform structure. However, when evaluating on more complex computation graphs such as matrix multiplication or FFT, we found that the above bound gave trivial results ($J^* \geq 0$) for every one of the  graphs. We hypothesize that the prescribed sub-graph size of $2*M$ is too small for more complex graphs . In our evaluation, we display results of the convex min-cut method run over the entire graph (without partitioning):
$$J_G^* \geq \max_{v \in V} \max(0, 2*(C(v,G) - M)).$$
The above bound is linear in $M$ for any graph. In the worst case, this bound can take $O(n^5)$ time to compute.


\subsection{Bound Behavior vs Graph Sizes}
\label{varyinggraphsize}

We examine graph behavior for varying graph sizes and varying $M$. We plot the the spectral method and the convex min-cut method against the graph size parameter ($l$ for the $2^l$ FFT, the matrix side length $n$ for matrix multiplication, and the number of cities for the TSP).  

To compare against the published bounds, we also plot the computed I/O for the spectral bound against the graph parameter term in the analytical bounds in Section \ref{specification}. For example, we plot the computed I/O for the $2^l$ point FFT graph against $l 2^l$.  If our bounds follow the growth patterns of the analytical bounds, then these plots should be roughly linear. We do not display points where the maximum in-degree is greater than $M$, because then the computation of some operations would not fit all their operands inside fast memory.

For all four graphs, we find that the bound computed from the spectral method is both tighter and more scalable than the convex min-cut method. In particular, the convex min-cut method is trivial for the naive matrix multiplication graph. 

Moreover, we find that our bounds roughly match the analytical growth of the published bounds, as the I/O vs the published bound is roughly linear for all four graphs. Finally, we note that our bound does not significantly degrade with $M$, and can thus be computed for large memory sizes.

\subsection{Scalability}
\label{scalability}

\begin{figure}[t]
    \includegraphics[width=0.5\textwidth]{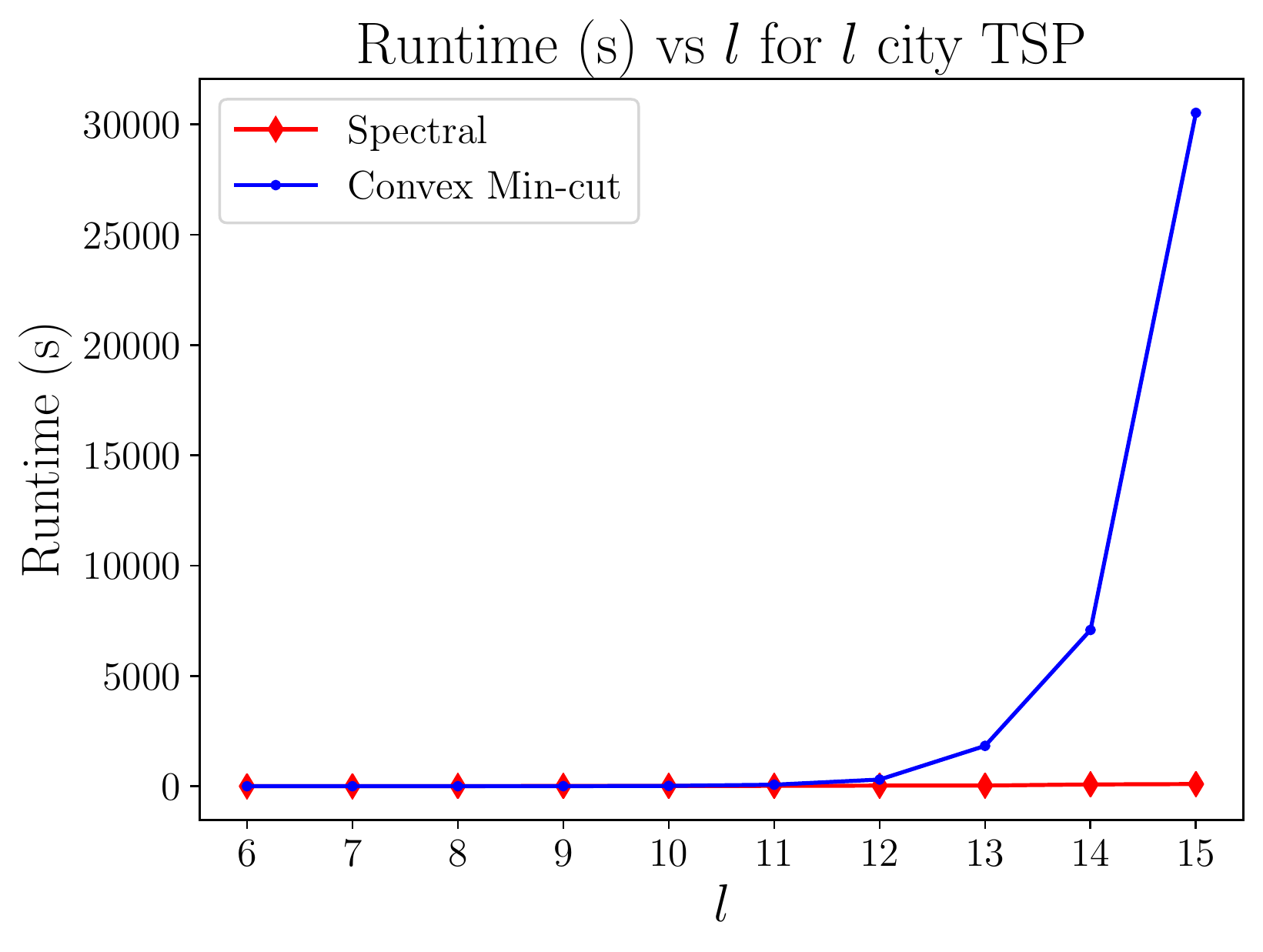}
    \caption{Runtime in seconds for computing the lower bound for Bellman-Held-Karp on a $l$ city TSP for various $l$.}
    \label{scalability_plot}
\end{figure}

Our spectral method is fast to compute, with a runtime complexity of $O(hn^2)$ where $h \leq n$ is the number of eigenvalues computed. Since any number of partitions $k$ in Theorem \ref{laplacianbound} gives a valid lower bound, $h$ can be set to trade off the bound strength with runtime complexity, with a maximum runtime of $O(n^3)$ by computing all the eigenvalues of the graph Laplacian. However, empirically we found that even when computing all of the eigenvalues, the best $k$ is usually far below $100$ even for large graphs, so the higher level eigenvalues remain unused: we therefore can set $h=100$ without losing bound strength. In contrast, the convex min-cut method has runtime complexity of $O(n^5)$, which does not scale well to large graphs. In Figure \ref{scalability_plot}, we plot the runtime in seconds of computing the convex min-cut and the spectral method for successively larger $l$ for evaluating Bellman-Held-Karp on an $l$ city TSP. We find that the convex-min-cut runtime explodes, taking close to 8.5 hours for the 15 city TSP, while our spectral method takes 98 seconds.

\section{Conclusion}
Finding I/O bounds for general computations remains a challenging problem. In this paper, we propose a novel method to find I/O bounds for computation graphs, using the spectra of the graph Laplacian. The spectra can be computed efficiently even for large graphs and can also be computed in closed form for some graphs, yielding a proof technique to find new closed-form bounds. We used the spectral method to derive closed-form bounds for several graphs, including the hypercube for the Bellman-Held-Karp algorithm and the butterfly graph for the Fast Fourier Transform. We evaluated our method empirically on four computation graphs and showed that it finds tighter bounds than previous automated methods at a fraction of the runtime and behaves similarly to published analytical bounds.

%
\begin{acks}
We thank Pratiksha Thaker, Moses Charikar, and Guillermo Angeris for their advice and feedback on this work. 

This research was supported in part by affiliate members and other supporters of the Stanford DAWN project---Ant Financial, Facebook, Google, Infosys, NEC, and VMware---as well as Cisco, SAP, and the NSF under CAREER grant CNS-1651570. Any opinions, findings, and conclusions or recommendations expressed in this material are those of the authors and do not necessarily reflect the views of the National Science Foundation. 

\end{acks}

%
\bibliographystyle{ACM-Reference-Format}
\bibliography{sources.bib}

\appendix

\section{Spectra of Butterfly Graphs}
\label{ffteigappendix}

We derive the Laplacian spectra of the $k$-level unwrapped butterfly graph $B_k$. $B_k$ has $(k+1)2^k$ vertices, which can be arranged as $k+1$ columns of $2^k$ nodes. Let $B_0$ be a single vertex. We form $B_k$ by connecting two disjoint copies of $B_{k-1}$.  Let $V_{i,j}$ be the $j$th vertex of the $i$th column of the first copy, and let $V'_{i,j}$ be similarly defined for the second copy. Via a new column of $2^k$ vertices (split into $V_{k+1, j}$ and $V'_{k+1, j}$ for $j=1,...,2^{k-1}$), we connect the two copies by adding edges $(V_{k, i}, V_{k+1,i})$, $(V'_{k, i}, V'_{k+1,i})$, $(V_{k, i}, V'_{k+1,i})$, and $(V'_{k, i}, V_{k+1,i})$. Examples of this labelling scheme are displayed in Figure \ref{fft_general}.

\begin{theorem}[Laplacian spectra of the butterfly graph]
    The graph Laplacian of the butterfly graph $B_k$ has eigenvalues:
    \begin{itemize}
        \itemsep0em 
        \item Repeated once:
        $4 - 4 \cos \left( \frac{\pi j}{k} \right), j=0,...,k$
        \item For $i=1,...,k$, repeated $2^{k-i+1}$ times:
        $4 - 4 \cos \left( \frac{\pi (2j+1)}{2i+1} \right), j=0,...,i-1$
        \item For $i=1, ..., k-1$, repeated $(k-i)2^{k-i-1}$ times:
        $4 - 4 \cos\left(\frac{j\pi}{i+1}\right), j=1,...,i$
    \end{itemize}
    \label{ffteigs}
\end{theorem}
\begin{proof}

    We decompose the butterfly graph into smaller, weighted graphs. We allow weights to exist both on edges and vertices via the weight functions $\omega: V \times V \rightarrow \bbR$ and $\phi: V \rightarrow \bbR$ respectively. If unmarked, edges have weight $1$ and vertices weight $0$. Then the Laplacian $L(G)$ of a vertex/edge weighted graph $G=(V,E)$ is
    \begin{align*}
        L_{ij}(G) = \begin{cases}
            \phi(v_i) + \sum_{(v_i, v_j) \in E} \omega(v_i, v_j) & i=j\\
            -\omega(v_i, v_j) & i\neq j
        \end{cases}
    \end{align*}   

    Let $\uplus$ indicate the multiset union. We define the following \textit{augmented butterfly graph} $B_k(q)$ for binary string $q \in \{0,1\}^m$ with $m \leq k$. $B_k(q)$ has $k+1$ columns of $2^{k-m}$ nodes. The first $k-m+1$ columns are the butterfly graph $B_{k-m}$. The next $m$ columns are connected as follows: for $i \in \{1,...,m\}$, if $q_i = 1$, create edges of weight $2$ $(V_{k-m+i, j}, V_{k-m+i+1, j})$ and $(V'_{k-m+i, j}, V'_{k-m+i+1, j})$ for $j \in \{1,...,2^{k-m-1}\}$. Otherwise if $q_i=0$, then add $2$ to the vertex weights of $V_{k-m+i, j}$, $V_{k-m+i+1, j}$, $V'_{k-m+i, j}$ and $V'_{k-m+i+1, j}$. Note $B_k([]) = B_k$, and augmentations of $B_k$ maintain vertex degrees. Under our formulation, $B_k([1])$ is equivalent to the augmented butterfly graph from \cite{avior1998tight}.  We prove the following decomposition lemma:
    \begin{lemma} Let $B_k(q)$ be an augmented butterfly graph with $|q| < k$. Then:
    $$\lambda(L(B_k(q))) = \lambda(L(B_k([1] + q))) \uplus \lambda(L(B_k([0] + q))).$$
    \vspace{-5mm}
    \label{decomp}
    \end{lemma}
    \begin{proof}
        Let $V_1$ be the vertices $V_{i,j}$ and $V_2$ be the vertices $V'_{i,j}$ for $i \in \{1,...,k+1\}$, $j \in \{1,...,2^{k-m-1}\}$ in $B_k(q)$. In Figure \ref{fft_augmented}, $V_1$ is the top half and $V_2$ is the bottom half of the vertices in each graph. The sub-graphs induced by $V_1$ and $V_2$ are identical. Let $D', A'$ be the degree and adjacency matrices of the sub-graph, and let $\hat{A}$ be the adjacency matrix between $V_1$ and $V_2$.  

        Due to the symmetry of the augmented butterfly, we can decompose the graph Laplacian into quadrants as 
    $$L(B_k(q)) = \begin{bmatrix} C_1 & C_2 \\ C_2 & C_1 \end{bmatrix}$$
        where $C_1 = D' - A'$ and $C_2 = -\hat{A}$. Then $C_1 + C_2 = D' - (A' - \hat{A}) = L(B_k([1] + q))$ and $C_1 - C_2 = D' - (A' - \hat{A}) = L(B_k([0] + q))$. Noting that $\lambda(L(B_k(q))) = \lambda(C_1 + C_2) \uplus \lambda(C_1 - C_2)$ \cite{bilu2006lifts} completes the proof of this lemma.
    \end{proof}

    \begin{figure}[t]
        \includegraphics[width=0.4\textwidth]{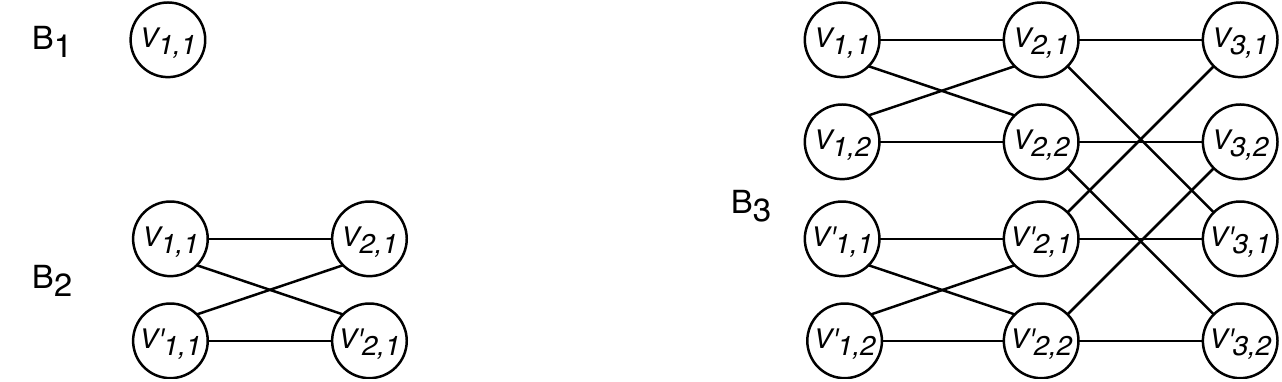}
        \caption{Three iterations of the unwrapped butterfly graph.}
        \vspace{-5mm}
        \label{fft_general}
    \end{figure}
    
    \begin{figure}
            \includegraphics[width=0.4\textwidth]{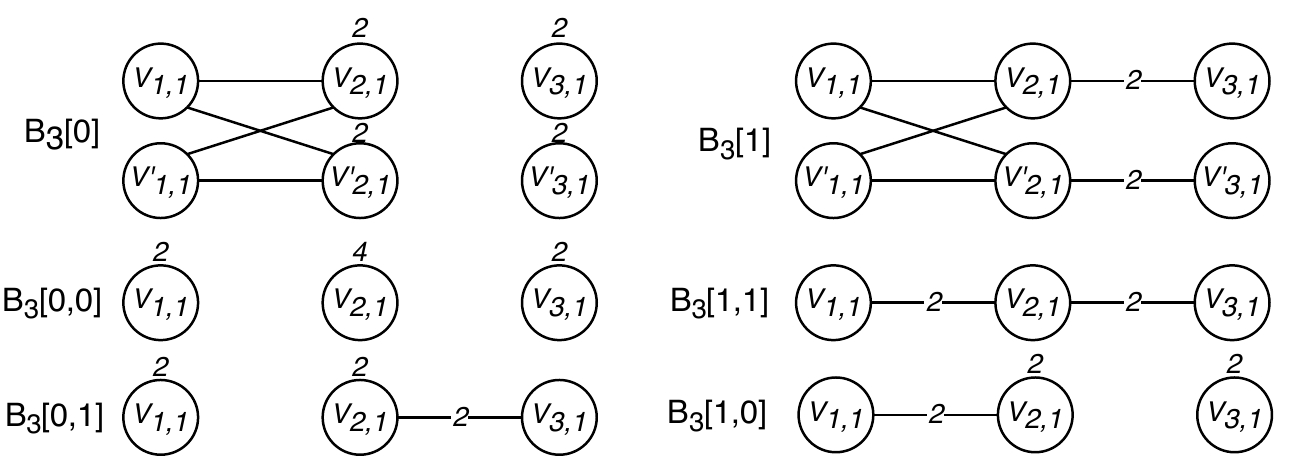}
        \caption{Examples of augmented butterfly graphs created from $B_3$.}
        \label{fft_augmented}
        \vspace{-5mm}
    \end{figure} 
    
    We inductively apply this lemma to decompose $B_k$ into a series of flat graphs. Define the graph $K(q)$ for $q \in \{0,1\}^k$ as follows: let $K(q)$ be the path graph of length $k+1$. Then for each $i \in \{1,...,k\}$, if $q_i = 0$,  delete the edge $(v_i, v_{i+1})$ and add 2 to the weights of the endpoints of that edge. Examples of $K(q)$ for $B_3$ can be found in Figure \ref{paths}. 

\begin{lemma}
    Let $Q = \{0,1\}^k$ be the set of all binary strings of size $k$. Then we have (counting multiplicity)
    $$\lambda(B_k) = \biguplus_{q \in Q} \lambda(L(K(q))).$$
\end{lemma}
\begin{proof}
    The proof follows from induction on Corollary \ref{decomp}. For permutation $q'$, let $Q = \{0,1\}^{k-|q'|}$. Then:
    $$\lambda(g(B_k(q'))) = \biguplus_{q \in Q} \lambda(g(B_k(q + q'))).$$
     For our base case, suppose that $|q'| = k$. Then $P$ is the emptyset, so the claim is trivially true. For our inductive step, assume that the claim is true for all $q'$ so that $|q'| = m$ (assuming $m>0$). Then we prove that the claim is true for $q'$ so that $|q'| = m-1$. Note that according to Corollary \ref{decomp}, 
    $$\lambda(g(B_k(q'))) = \lambda(g(B_k([0] + q')))  \uplus \lambda(g(B_k([1] + q'))).$$
    Let $Q' = \{0,1\}^{k-|q'|-1}$.
    But by the inductive hypothesis, we have that:
    $ \lambda(g(B_k([0] + q'))) = \biguplus_{q \in Q'} \lambda(g(B_k(q + [0] + q')))$ and
    $\lambda(g(B_k([1] + q'))) = \biguplus_{q \in Q'} \lambda(g(B_k(q + [1] + q'))).$
    Then combining the two together we get 
    $\lambda(g(B_k(q'))) = \biguplus_{q \in Q} \lambda(g(B_k(q + q'))).$
    Thus our inductive claim holds. $q' = []$ is a special case of this claim that completes the proof.
\end{proof}

Each $K(q)$ consists of disconnected path graphs. Let $P_i$ be the path graph with $i$ vertices with edge weights 2. $P_i'$ has one end vertex with weight $2$, and path $P_i''$ has both end vertices with weight 2. Examples of these paths can be found in \ref{path_types}.

\begin{figure}[t]
    \includegraphics[width=0.5\textwidth]{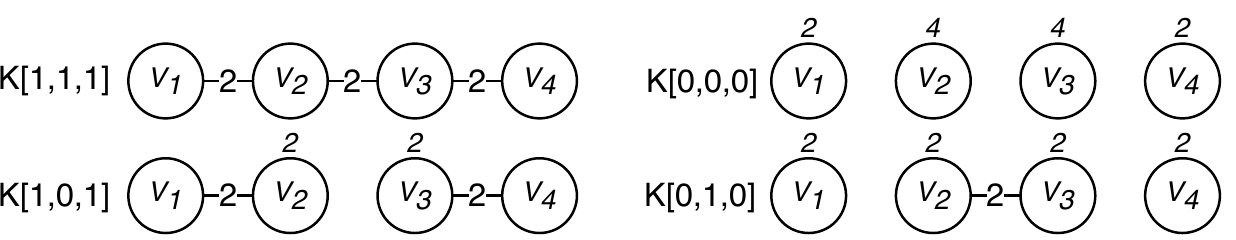}
    \caption{Examples of $K(q)$ for various $q$.}
    \vspace{-5mm}
    \label{paths}
\end{figure}

\begin{figure}
    \includegraphics[width=0.5\textwidth]{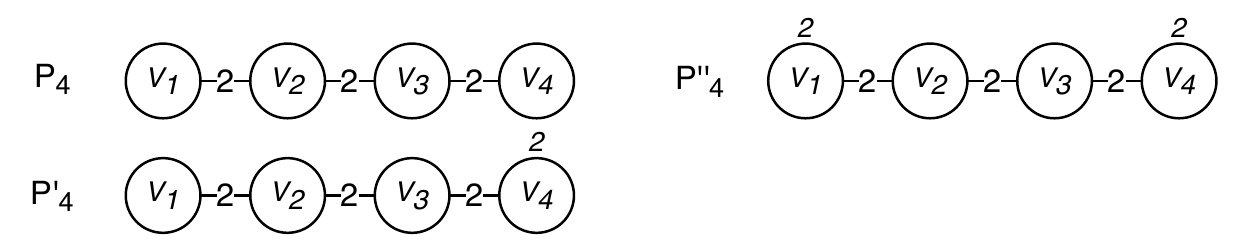}
    \caption{Three types of paths: $P_4$, $P'_4$, $P''_4$ in $\mathcal{K}$ for $B_4$.}
    \vspace{-5mm}
\label{path_types}
\end{figure}

\begin{lemma}
    $$\lambda(L(B_k)) = \biguplus_{P \in \mathcal{P}} \lambda(L(P))$$
    where $\mathcal{P}$ consists of:
    \item[i.] A single instance of $P_{k+1}$
    \item[ii.] For $i=1,...,k$, $2^{k-i+1}$ instances of $P_i'$.
    \item[iii.] For $i=1,...,k-1$, $(k-i)2^{k-i-1}$ instances of $P_i''$
\end{lemma}
\begin{proof}
    Define $\mathcal{K}$ as the multiset of path components of each $K(q)$ built from all $q \in Q = \{0,1\}^k$:
    $\mathcal{K} = \biguplus_{q \in Q} K(q).$ We need to prove that $\mathcal{P} = \mathcal{K}$.

    \textbf{(i.)} The only instance of $P_{k+1}$ is created by $K(\{1\}^k).$

    \textbf{(ii.)} We first examine instances of $P'_{i}$ such that the weighted vertex is on the right. Since the edge immediately after the first $i-1$ edges must be deleted, there are $k-i$ edges to the left that can be either 0 or 1. Then there are $2^{k-i}$ paths of $P'_{i}$ that appear such that the weighted vertex is on the right. Since the number of paths with the weighted vertex as the left or the right is symmetric, there are a total of $2^{k-i+1}$ instances of $P'_{i}$.

    \textbf{iii.} We examine an instance of $P''_{i}$ that starts on a specific index. The edges on both the left and the right of the path must be deleted. Therefore, there are $k-i-1$ edges that are free to be either 0 or 1. There are $2^{k-i-1}$ instances of $P''_{i}$ which begin at a specific instance. $P''_{i}$ can begin at $k-i$ possible valid indices, resulting in $(k-i)2^{k-i-1
    }$ instances.

    Thus $\mathcal{K} = \mathcal{P}$. Lemma \ref{decomp} completes the proof.
\end{proof}

Thus, to finish our proof of Theorem \ref{ffteigs}, we simply need to find the eigenvalues of $P_i$, $P_i'$, and $P_i''$.

\begin{lemma}
\end{lemma}
The eigenvalues of $L(P_i)$, $L(P_i')$ and $L(P_i'')$ are:
\begin{itemize}
\item $\lambda(L(P_i)) = 4 - 4 \cos \left( \frac{\pi j}{i} \right), \forall j=0,...,i-1$
\item $\lambda(L(P_i)') = 4 - 4 \cos \left( \frac{\pi (2j+1)}{2i+1} \right), \forall j=0,...,i-1$
\item $\lambda(L(P_i'')) = 4 - 4 \cos(\frac{j\pi}{i+1}), \forall j=1,...,i.$
\end{itemize}
\begin{proof}
    The eigenvalues of the unweighted path $P_i$ have been well-studied in current literature and can be found in \cite{brouwer2011spectra, chung1997spectral}. We omit the proof here.

    We note that $L(P_i')$ is a tri-diagonal matrix: it has $4$ on the diagonal, $-2$ on the off-diagonals, and $2$ in the lower right corner. We show that $\lambda(L(P_i'))$ are the odd eigenvalues of $\lambda(L(P_{2i+1}))$. 
    We can decompose $L(P_{2i+1})$ as:

    \begin{center}
    \includegraphics[width=0.3\textwidth]{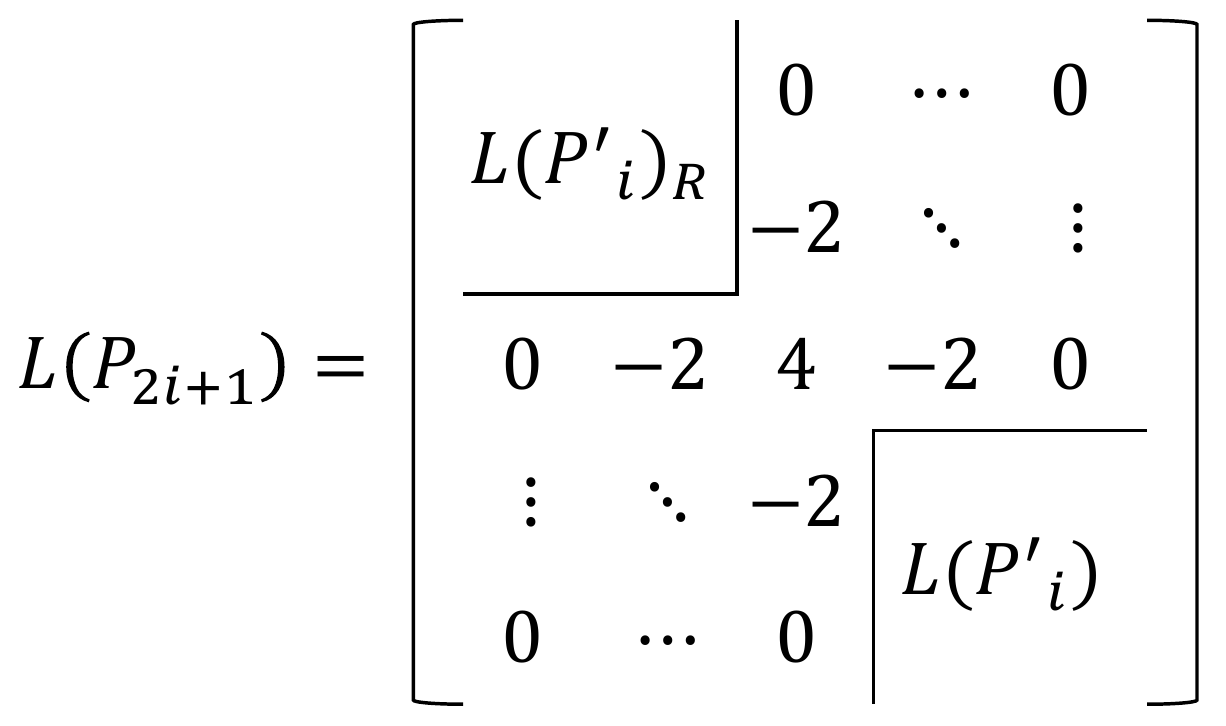}
    \end{center}
    where $L(P'_i)_R$ is just $L(P'_i)_R$ with reversed index order such that the first vertex has extra weight. Suppose we found an eigenvector $\mathbf{x}$ of $L(P_{2i+1})$ corresponding to eigenvalue $\lambda$ with the form:
    \begin{equation}
        \label{expandededeig}
    \mathbf{x} = \begin{bmatrix} \mathbf{y}_R \\ 0 \\ -\mathbf{y}\end{bmatrix}
    \end{equation}
    where $\mathbf{y}_R$ is the reversed vector of $\mathbf{y}$. Then, 
    $$L(P_{2i+1}) \mathbf{x} = \begin{bmatrix} L(P'_i)_R \mathbf{y}_R \\ 0 \\ -L(P'_i)_R \mathbf{y} \end{bmatrix} =  \lambda \begin{bmatrix}  \mathbf{y}_R \\ 0 \\ - \mathbf{y} \end{bmatrix}.$$
    Therefore $L(P'_i) y = \lambda y$ so $y$ is an eigenvector of $L(P'_i)$. Thus, if we can find $i$ orthogonal eigenvectors of $L(P_{2i+1})$ of the form Equation \ref{expandededeig}, we can reduce each of those eigenvectors to find all the eigenvectors of $P_i'$.
    
    We show that only the odd eigenvectors of $P_{2i+1}$ fall into the form of Equation \ref{expandededeig}. The eigenvector formulas for $L(P_{l})$ are (for $j=0, ..., l-1$) \cite{perera2012bipartition}:
    $ L(P_l): x_j(h) = 2\cos(\frac{(2h - 1)j\pi}{2l}), h=1,...,l
    $. Let $l=2i+1$, which means that $l$ is odd. By inspection, we note that if $h$ is odd and $l$ is odd, then both $x_j(h)$ and $\tilde{x}_j(h)$ follow the form in Equation \ref{expandededeig}. That means that the eigenvalues of $L(P_i')$ are the odd eigenvalues of $L(P_{2i+1})$. Plugging into the formulas for $\lambda(L(P_{2i+1}))$ finds the above closed for for $\lambda(L(P_i'))$.

    Finally, $L(P_i'')$ is Toeplitz with $4$ on the diagonals and $-2$ on the off-diagonals. The eigenvalues of tridiagonal Toeplitz matrices have a closed form \cite{noschese2013tridiagonal}, thus finishing the proof. 
\end{proof}
\end{proof}

\end{document}